\documentclass[10pt,journal]{IEEEtran}

\ifCLASSOPTIONcompsoc

\usepackage[nocompress]{cite}
\else

\usepackage{cite}
\fi

\ifCLASSINFOpdf

\else

\fi

\usepackage{subfigure}
\usepackage{graphicx}

\usepackage{amsmath,amsthm}
\usepackage{amssymb,wasysym}
\usepackage{amsfonts,balance}
\usepackage{mathrsfs}
\usepackage{cite}
\usepackage{color}
\usepackage{url}
\usepackage{bm}
\usepackage{wrapfig}
\usepackage{verbatim}
\usepackage{dsfont}
\usepackage{amsmath}
\usepackage{algorithm}
\usepackage{algpseudocode}
\usepackage{xurl}
\usepackage{tabularx}
\usepackage{enumitem}
\usepackage[colorlinks=true, linkcolor=black, citecolor=magenta, urlcolor=blue, bookmarks=false]{hyperref}

\newtheorem{theorem}{\textit{Theorem}}
\newtheorem{lemma}{\textit{Lemma}}

\newtheorem{remark}{\textit{Remark}}
\allowdisplaybreaks[3]

\ifodd 1
\else

\fi

\graphicspath{{figures/}}

\begin{document}

	\title{Inter-Satellite Link-Enhanced Transmission Scheme Towards Aviation IoT in SAGIN}
	
	\author{Qian Chen, \IEEEmembership{Member, IEEE}, Chenyu Wu, \IEEEmembership{Member, IEEE}, Shuai Han, \IEEEmembership{Senior Member,~IEEE}, \\
	Weixiao Meng, \IEEEmembership{Senior Member,~IEEE}, and Tony Q. S. Quek, \IEEEmembership{Fellow,~IEEE}
	
	\thanks{
		An earlier version of this work \cite{chen2024exploiting} has been accepted by the 2024 IEEE Global Communications Conference (GLOBECOM).  (\textit{Qian Chen and Chenyu Wu are co-first authors.)}	\textit{(Corresponding author: Weixiao Meng.)}
		
		Q. Chen,  C. Wu, S. Han and W. Meng are with the School of Electronics and Information Engineering, Harbin Institute of Technology, Harbin 150001, China (e-mail: joycecq@163.com; wuchenyu@hit.edu.cn; hanshuai@hit.edu.cn; wxmeng@hit.edu.cn).
		
		T. Q. S. Quek is with the Information Systems Technology and Design (ISTD), Singapore University of Technology and Design, Singapore 487372 (e-mail: tonyquek@sutd.edu.sg).

	}
}

	\markboth{}%
	{Shell \MakeLowercase{\textit{et al.}}: Bare Advanced Demo of IEEEtran.cls for IEEE Computer Society Journals}

	\IEEEtitleabstractindextext{
		\begin{abstract}
				The rapid development of the aviation Internet of Things (IoT) has positioned in-flight connectivity (IFC) as one of its critical applications. Space-air-ground integrated networks (SAGIN) are essential for ensuring the performance of IFC by enabling seamless and reliable connectivity. However, most existing research treats satellites merely as transparent forwarding nodes and overlooks their potential caching capabilities to enhance IFC data rates. In this article, we explore an IFC-oriented SAGIN where satellites and ground stations (GSs) work together to transmit content to airborne passengers, thereby facilitating airborne communication. By categorizing files into cached (instantly accessible via satellites) and non-cached files (available only through GSs), this article pioneers the integration of \textit{multiple inter-satellite links} (ISLs) into the IFC framework, thus innovating the content delivery process for both types of files. To minimize the average delay of content delivery, we formulate the corresponding optimization problems: 1) For cached files, we propose an exact penalty-based method to determine the satellite association scheme. 2) For non-cached files, we present an efficient algorithm based on alternating optimization to jointly optimize satellite association and GS bandwidth allocation. Our proposed framework is low in complexity, paving the way for high-speed Internet connectivity for aviation passengers. Finally, simulation results are provided to demonstrate the effectiveness of our proposed IFC framework for SAGIN.
		\end{abstract}

		}

	\maketitle

	\IEEEdisplaynontitleabstractindextext

	%
	\IEEEpeerreviewmaketitle
	
	\section{Introduction}
	In recent years, the demand for ubiquitous and reliable internet connectivity has extended from terrestrial networks to the aviation sector \cite{8710357}. This shift has driven the development of the aviation Internet of Things (IoT), which integrates IoT technology into aerospace applications to support key goals like operational efficiency, passenger safety, in-flight experience, and predictive maintenance. Studies show that 82\% of respondents now view in-flight connectivity (IFC) as a fundamental expectation \cite{immersat_IFC}. Furthermore, introducing in-flight Wi-Fi is projected to add approximately \$4 in revenue per passenger, potentially contributing an additional \$30 billion to airline profits by 2035 \cite{oneZero_IFC}. By connecting passengers, onboard systems, and various airborne devices, aviation IoT plays a pivotal role in advancing both user-centric services and the overall effectiveness of airline operations.

	Thanks to the promising framework of the sixth-generation (6G) network, space-air-ground integrated networks (SAGIN) make reliable internet access on airplanes feasible \cite{10579820,wu2024enhancing}. When the content requested by passengers is not stored onboard, two primary methods are employed for delivery. The first method is through \textit{air-to-ground (A2G) communications}, where the aircraft connects directly to aircraft gateways. A2G communications typically offer low latency due to short transmission distances \cite{10024899}. 
	However, deploying aircraft gateways in remote areas, such as over oceans, is challenging, and their coverage is smaller than that of satellites, limiting A2G's practicality for ubiquitous IFC \cite{9511638}. 
	The second one is through \textit{air-to-space (A2S) communications}, where the aircraft connects to the nearest satellite, which relays the content from satellite ground stations (GSs). Although satellites provide extensive coverage, A2S communications experience higher transmission delays due to greater distances \cite{9565322}. To address this, popular files can be proactively stored on satellites, allowing direct delivery to airplanes via inter-satellite links (ISLs), bypassing GSs. However, current IFC schemes using A2S communications often overlooked the cache capabilities of satellites, assuming that all files are retrieved from GSs. During peak times, the burden on GSs and latency of G2S links can increase significantly. Furthermore, radio frequency (RF)-based ISLs face regulatory challenges and low data rates \cite{trends}. 
	Recently, optical wireless communication through laser ISLs has emerged as a promising solution for low-Earth orbit (LEO) satellite constellations \cite{9393372}, offering much higher speeds of 2.5 Gb/s and potentially reducing transmission latency within satellite networks.
	
	In this context, developing an efficient content delivery scheme for IFC is essential to meet the growing data rate demands of airborne internet access. This article leverages the advantages of laser ISLs and proposes a comprehensive framework for achieving IFC within SAGIN. Compared to existing methods, the proposed ISL enhanced IFC scheme can reduce content delivery latency and alleviate the communication load on feeder links.

	\subsection{Inter-Satellite Links}
	When implementing ISLs for content delivery, it is essential to account for the limited number of ISLs a satellite can establish at any given time, as well as the data volume originating from multiple links.
	Since the mega LEO constellation networks increase the content delivery complexity and the required ISL hop count, Chen \textit{et al.} proposed a theoretical model to estimate the hop count between ground users when leveraging satellite networks for data transmission \cite{9351765}.
	The authors extended their work to minimize the total ISL usage by optimizing the gateway placement and studied the relationships between ISL hop count and gateway site \cite{9794736}.

	Previously, RF-based links were used for long-distance communication in space. Thanks to the technological development of optics and laser techniques \cite{10644103}, the era of using Terahertz (THz) and laser beams for inter-satellite communication is coming, making high-speed transmission in space possible \cite{9448978}. SpaceX has indicated that each Starlink satellite can establish up to four laser ISLs \cite{Space-ISL}. These links are strategically set up to connect each satellite with two others in the same orbital plane (OP) and an additional two in different OPs. 
	With the advantages of energy efficiency and reduced hardware complexity, reconfigurable intelligent surfaces (RISs) are anticipated to play a crucial role in the integration of THz waves in small satellites \cite{9772693,10274508}.
	Zheng \textit{et al.} proposed a RIS-aided LEO satellite communication system and jointly optimized the active transmit/receive beamforming at the satellite and ground node to maximize the total channel gain \cite{9849035}.
	Tekbiyik \textit{et al.} evaluated the application of THz band in ISLs and measured how misalignment fading affects error performance \cite{9954397}. To enhance signal-to-noise ratio (SNR), they proposed to employ RIS installed on adjacent satellites to facilitate signal transmission. 
	
	Applying ISLs can also construct a decentralized network in space and thus reduce the requirement for a large number of GSs.
	Yan \textit{et al.} proposed an optimization method to achieve the optimal latency and throughput between satellites and GSs simultaneously when ensuring the ranging performance \cite{9211777}. 
	Qi \textit{et al.} focused on the ISLs between satellites in different OPs and proposed an ISL scheduling method to achieve a tradeoff between the system complexity and latency performance \cite{9475443}. Nonetheless, existing research has not adequately addressed the impact of multi-satellite cooperative transmission via laser ISLs for content delivery in SAGIN.

	\subsection{In-Flight Connectivity in SAGIN}
	The dynamic nature of satellites and aircraft makes content delivery and routing more intricate in SAGIN than terrestrial networks. In terrestrial networks, content delivery or handover is typically triggered by user movement, and direct links between small base stations and core networks are consistently available \cite{qu2024trimcaching,10286277}.
	For IFC in SAGIN, both the aircraft and the satellites are moving. Moreover, the potential link interruptions introduce uncertainties in routing \cite{10675381,9984697,10121438}.
	
	A2G offers a solution to the IFC's high propagation delay associated with satellite links \cite{3GPP_WID_ATG}. However, geographical factors often restrict gateways' placement, making A2G links discontinuously feasible during a flight.
	To tackle this problem, an intercontinental backbone (ICB) can be established by selecting aircraft that travel between specific source and destination airports on different continents \cite{1688020}. 
	However, creating such an aeronautical ad-hoc network (AANET) often requires multi-hop routing when direct A2G links are not possible. Conflicts of interest may arise among airlines when multiple aircraft are within the ICB, leading to security issues in such a distributed system \cite{wu2022echohand,liang2024ponziguard}.
	
	Leading companies in the LEO satellite sector, such as SpaceX and OneWeb, have recently incorporated IFC as one of their crucial business branches. SpaceX has announced that aviation customers using their Starlink service can expect a data speed of up to 350 Mbps, which provides all passengers with the capability for high-speed internet streaming.
	For academic studies, most existing works have assumed that the requested services are downloaded from a single satellite \cite{9322493,9361631} and designed one-to-one matching association schemes. 
	For instance, to handle the complexity of integer linear programming, an efficient online algorithm was proposed to solve the in-flight service provisioning problem in polynomial time \cite{9758064}. Another study introduced a flight control framework based on deep Q-learning to manage airplane-to-data center assignments, routing, and reconfiguration decisions during the flight period \cite{9322493}.
	These works generally considered a binary migration strategy, meaning that once the association objective is determined, all files are delivered from the same node. \cite{9361631} was the first work to consider a partial migration strategy in aircraft data access within SAGIN. The authors jointly optimized the data downloading ratio from AANET and ground stations, offering a general approach to addressing the multiple-to-one matching association issues for IFC.

	To simultaneously utilize the broad coverage of LEO satellites and the low round-trip time (RTT) of aircraft gateway, it is necessary to develop an adaptive association strategy, which integrates the network topology and quality of service (QoS) requirements. 
	Currently, the field of dual connectivity for IFC is relatively nascent. The work \cite{8715404} and \cite{8761265} begun exploring the traffic scheduling challenges in aeronautical networks within such a dual connectivity framework in SAGIN. In these works, the aircraft can obtain the services from the aircraft gateway or the satellite ground stations via the relay satellites. 
	However, these studies merely treated satellites as transparent forwarding nodes, overlooking their potential capabilities for service caching. As intelligent and flexible LEO satellite networks become prevalent, LEO satellites are increasingly capable of communicating, computing, and caching simultaneously \cite{9978924, 2023arXiv231101483L}. 
	With a direct A2S link, latency can be further reduced since the downloading latency from ground stations to satellites can be reduced. This highlights the need for evolving research and strategies in dual connectivity to fully leverage the capabilities of advanced LEO satellite networks for IFC.

	\subsection{Our Contributions}
	This paper studies the content delivery strategy in SAGIN for achieving IFC. The main contributions of this paper are as follows:
	\begin{itemize}
		\item \textbf{Propose an efficient cooperative transmission framework by exploiting laser ISLs for realizing uninterrupted IFC}. By employing the caching capabilities of satellites, we categorize the files into two types: those that have been cached at satellites (immediately reachable from satellites) and those that have not (only achievable through GSs). 
		Based on the verification that the link packet error rate (PER) falls within the reasonable range, we formulate an optimization problem for each type of files to minimize the average delay tailored to the inherent features of laser ISLs. To the best of our knowledge, \textit{this is the first paper to investigate the enhancement effects of multiple ISLs on IFC to minimize the content delivery delay within SAGIN}.
		
		\item \textbf{Design optimization algorithms to tackle the mixed integer programming problems.}
		Regarding the delay minimization problem of the cached files, we analyze the property of the optimal solution and transform the original mixed integer nonlinear programming (MINLP) problem into a more tractable form with only binary variables regarding satellite association. We then propose an exact penalty-based algorithm to deal with it with low complexity. For the more challenging non-cached case, we devise a new alternating optimization framework by efficiently grouping the association and bandwidth allocation variables.
		
		\item \textbf{Verify the effectiveness of the proposed framework.}
		Simulation results are provided to validate the effectiveness of the proposed framework as compared to several benchmarks. Moreover, we evaluate the impacts of various key parameters on the performance of the proposed algorithms. Notably, increasing the maximum number of ISLs can significantly reduce the content delivery delay, which verifies the necessity of exploiting multiple ISLs for IFC.

	\end{itemize}
	
	The rest of this paper is organized as follows. In Section \ref{sec_sm}, we introduce the network model, file model, communication model, and PER of the considered SAGIN for IFC. In Section \ref{sec:cached} and \ref{sec:noncached}, we investigate the delivery scheme of the cached files and non-cached files, respectively. Finally, simulation results and conclusions are provided in Section \ref{sec_sr} and \ref{sec_conclusion}. 
	The main notations in this paper are summarized in Table \ref{tab:notation}.

	\begin{table*}[ht]
		\caption{Summary of main notations} \label{tab:notation}
		\vspace{-10pt}
		\begin{center}
			{\footnotesize	\begin{tabular}{|c|p{11cm}|}\hline  
					\textbf{Notation} & \textbf{Definition}  \\ \hline
					$ \mathcal{A}, \mathcal{S}, \mathcal{G} $ &  The set of aircraft, satellites, and aircraft gateways or satellites GSs. \\ \hline
					$ a_i^t, s_j^t, s_k^t, g_m  $ & The aircraft with content request, the nearest satellite of aircraft, the nearest satellite's neighboring satellite, and gateway/GS. \\ \hline
					${\mathcal{L}_{\rm{G2A}}},  {\mathcal{L}_{\rm{G2S}}},  {\mathcal{L}_{\rm{S2A}}},  {\mathcal{L}_{\rm{ISL}}}$ & The set of G2A, G2S, S2A links, and ISLs. \\ \hline
					$ \mathcal{T}, t, \tau $ & The set of TSs, the index of TS, and duration of each TS. \\ \hline
					$ \mathcal{F}, f $ & The set of requested files, and the index of file. \\ \hline
					$\mathbf{X}, \bm{\rho}$ & Satellite association matrix and download ratio matrix of the \textit{cached} file schemes. \\ \hline
					$\tilde{\bf X}, {\tilde {\bm{\rho}}}, {\bm{\omega}}$ & Satellite association matrix, download ratio matrix, and bandwidth allocation matrix of the \textit{non-cached} file schemes. \\ \hline
					$R_p, b_f$ & The number of bits contained in each data packet, and the number of data packets contained in file $f$. \\ \hline
					$N_{{\rm{ISL}},\max }$ & The maximum number of ISLs which each satellite can establish within a given time slot. \\ \hline
					$ C\left( {s_k^t,s_j^t} \right), C\left( {g_m,s_k^t} \right)$ & Achievable rate of ISL  and G2S links.  \\ \hline
					$t_f^{\rm{tran}}\left( \cdot \right), t_f^{\rm{prop}}\left( \cdot \right)$ & Transmission and propagation latency of different communication links when delivering file $ f $. \\ \hline
					$ {{t_{f,{\rm{ISL}}}^{\rm{cached}}}\left( {s_j^t} \right),{t^{\rm cached}_{f,{\rm{G2S}}}}\left( {g_m,s_j^t} \right)} $ & Latency through all ISLs and G2S link during Phase 1 when delivering the \textit{cached} contents from the nearest satellite $ s_j^t $. \\ \hline
					$t_{f,{\rm{SN}}}^{{\rm{non}},{\rm{case1}}}\left( {s_j^t,a_i^t} \right), t_{f,{\rm{SN}}}^{{\rm{non}},{\rm{case2}}}\left( {s_j^t,a_i^t} \right)$ & Total latency of two cases when delivering \textit{non-cached} files from satellite networks, respectively. \\ \hline
					$ t_f^{{\rm{relay}}}\left( {{g_m},s_k^t,s_j^t} \right) $ & The two-hop delay when satellite $s_j^t$ downloads \textit{non-cached} file $f$ from GS $g_m$ via neighboring satellite $s_k^t$. \\ \hline
					${{t_{f,{\rm{S2A}}}}\left( {s_j^t,a_i^t} \right)}$  & Latency through S2A link during Phase 2. \\ \hline
					${t_f^{\rm{cached}}}\left( {a_i^t} \right), {t_f^{\rm{non-cached}}}\left( {a_i^t} \right) $ &  Total delay of aircraft $a_i^t = {\rm{sn}}_f$ when downloading files $f$ which are \textit{cached} and \textit{non-cached} at satellites.  \\ \hline
					$ {\mathbf{u}}, {\mathbf{v}}, {\bm{\pi}}, S$ & Auxiliary variables when solving $\mathcal{P}1$. \\ \hline
					$l, r$ & Iteration number when solving $\mathcal{P}4$  and $\mathcal{P}5$. \\ \hline
			\end{tabular}}
		\end{center}
	\end{table*}
	
	\section{System Model}\label{sec_sm}
	\subsection{Network Model}\label{subsec:network_model}
	We consider an IFC-oriented SAGIN as illustrated in Fig. \ref{fig:SAGIN_Architecture}, where the space and ground network cooperatively deliver content to the passengers in aircraft to realize in-flight connectivity. Considering the dynamic nature of the network topology, we adopt the finite state automation (FSA) model \cite{FSA}, where the whole schedule period is divided into a series of equal-length time slots (TSs), each with a duration of $ \tau $, indexed by the set $ {\mathcal T} = \left\{1,\ldots, T\right\} $. The network topology remains relatively stable within each TS but may change between different TSs.
	
	\begin{figure}[ht]
		\centering
		\includegraphics[width = 0.45\textwidth]{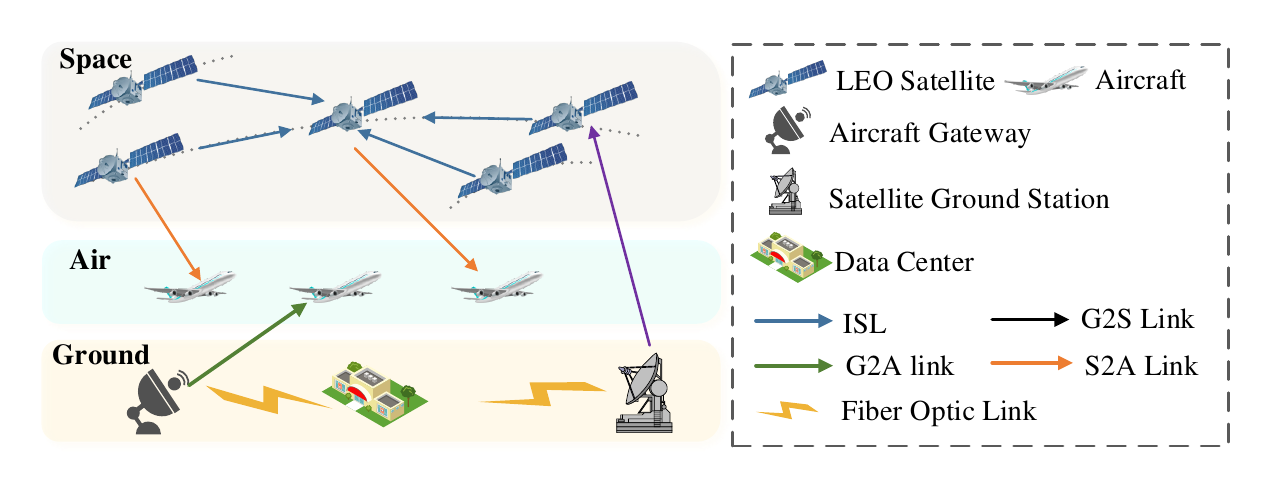}
		\caption{Framework of IFC-oriented SAGIN. \label{fig:SAGIN_Architecture}} 
	\end{figure}
	
	In our considered IFC-oriented SAGIN, $ A $ aircraft and $ S $ satellites follow their predefined flight paths or routes. To describe the movement of these non-terrestrial platforms during TS $ t \in \mathcal{T} $, the aircraft are represented as $ {\mathcal A} = \left\{a_1^t, \ldots, a_A^t \right\} $, and the satellites are represented as $ {\mathcal S} = \left\{s_1^t, \ldots, s_S^t \right\} $. 
	In addition, $G$ aircraft gateways and satellite GSs are deployed to be connected to the core network (data center) through high-speed wired links, where signals are transmitted via optical communication. As the aircraft gateways and satellite GSs are static, they are denoted by $ {\mathcal{G}} = \left\{g_1, \ldots, g_G \right\} $.  Given appropriate site selection, the transmission delay from the data center to the gateways or GSs is negligible. As such, it is reasonable to assume that the gateways and GSs have access to all the required files.
	Each aircraft can connect to at most one satellite or one aircraft gateway at a time due to the limited antenna deployed in cabin.
	Each gateway and GS has multiple antennas, but only one antenna associates with a satellite at a time \cite{2023mobi}. Thus, multiple satellites may attend for the bandwidth of the same gateway or GS.

	\subsection{File Model}
	While the aircraft is in cruise phase, airborne passengers periodically make multimedia requests, e.g., streaming videos and browsing websites. According to the storage status of the file request, all files can be divided into the following two categories: 1) \textbf{Cached contents}, such as popular movies, music, and books. They can be proactively cached in the satellite network according to the statistics of user preferences. As such, they can be shared among satellites via laser ISLs. Since satellites have limited storage capacity, these files can be downloaded from the GSs when they are not cached on satellites. Then it becomes the following non-cached case. \footnote{When the requested files have already been stored in the aircraft cabin, they can be delivered directly to the passengers. We neglect in this paper this simple case and focus on the procedure of acquiring data from external nodes of the aircraft, i.e., satellites and GSs.}  2) \textbf{Non-cached contents}, such as text messages, websites, and newsletters. These files generally can not be cached in non-terrestrial platforms since they are personally preferred or require real-time updating. In this context, the aircraft must download the requested files from the aircraft gateway or satellite GSs, while satellites act as relay nodes.
	Fig. \ref{fig:compare_sat_GS} compares the differences between downloading files from satellites and GSs regarding cache capability and transmission rate. The delivery schemes of these two categories of files will be detailed in Section \ref{sec:cached} and \ref{sec:noncached}, respectively. 
	\begin{figure}[h]
		\centering
		\includegraphics[width = 0.35\textwidth]{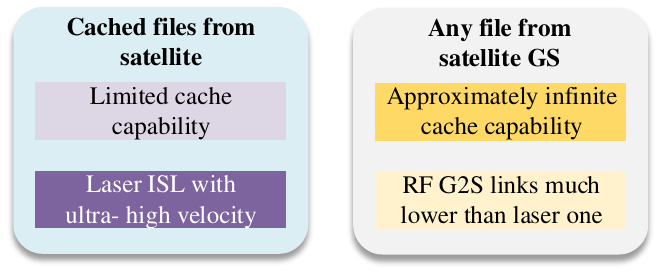}
		\caption{Comparison between file downloading from satellites and GSs. \label{fig:compare_sat_GS}}
	\end{figure}

	The requested files are represented by $ {\mathcal F} =  \left\{1,\ldots, F\right\}$. Each file $f \in {\mathcal F}$ can be described by a tuple $ \left( {\rm{sn}}_f, {\rm{gt}}_f, {\rm{tp}}_f, b_f\right)  $, where $ {\rm{sn}}_f \in \mathcal{A}$ is the source node that generates file requests, $ {\rm{gt}}_f \in \mathcal{T} $ is the generation time of the file request, and $ {\rm{tp}}_f $ contains the key information of the file (including whether it is cached and the specific type of the file, e.g., music, picture, movie, etc.), and $ b_f $ is the size in terms of the number of packets. Specifically, we follow a data-partition model, where each file can be divided into multiple packets of equal size, each containing $ R_p $ bits. This packet size represents the smallest unit of a task that can not be further divided.

	\subsection{Communication Model}\label{subsec:communication_model}
	Let $\cal L$ denote the union of all communication links. In our considered IFC-oriented SAGIN, there are four types of communication links in $\cal L$: 1) Ground-to-air (G2A) links, denoted as ${\mathcal{L}}_{\rm{G2A}}$\footnote{When a direct G2A link can be established, files can be downloaded directly from the aircraft gateways, without the need of satellite networks.}. 2) Ground-to-space (G2S) links ${\mathcal{L}}_{\rm{G2S}}$. 3) Space-to-air (S2A) links ${\mathcal{L}}_{\rm{S2A}}$. 4) ISLs ${\mathcal{L}}_{\rm{ISL}}$.  Due to the significant propagation delay of the links in space and air, for simplicity, the links in ${\mathcal{L}}_{\rm{G2A}}$, ${\mathcal{L}}_{\rm{G2S}}$, and ${\mathcal{L}}_{\rm{S2A}}$ are determined by the nearest distance method. In other words, aircraft connect to their nearest satellites when downloading files from the satellite network. Aircraft and satellites also set up links with their nearest visible gateways. Then, the establishment of ISLs and the achievable rate of each communication link in $\cal L$ are detailed as follows. 
	\subsubsection{ISL}
	Given that the laser (THz) band possesses enormous spectral resources, we assume orthogonal frequency for each ISL. Hence, let $ W_{\rm ISL} $ denote the common bandwidth of laser ISLs. When satellite $s_k^t$ transmits file $ f $ to satellite $s_j^t$ via ISL $l_f(s_k^t, s_j^t)$, the SNR at satellite $s_j^t$ is expressed as
	\begin{equation}\label{equ:SNR}
		\begin{split}
			&	\gamma \left( {s_k^t,s_j^t} \right) = \\
			&	 \frac{{{P_{\text{T}}}\left( {{s_k}} \right){G_{\text{T}}}\left( {{s_k}} \right){G_{\text{R}}}\left( {{s_j}} \right){L_{{\rm{add-ISL}}}}{L_{{\rm{PL}}}}\left( {s_k^t,s_j^t} \right){h^2}\left( {s_k^t,s_j^t} \right)}}{{{n_0}{W_{{\text{ISL}}}}}},
		\end{split}
	\end{equation}
	where $ {P_{\rm{T}}}\left( {{s_k}} \right) $ represents the transmit power of satellite $s_k$, $ {G_{\rm{T}}}\left( {{s_k}} \right)$ and $ {G_{\rm{R}}}\left( {{s_j}} \right)$ stand for the transmit and receive antenna gain at the two ends of the ISL, $ {L_{{\rm{add-ISL}}}}$ accounts for additional losses due to various factors regarding the communication environment, $ n_0 $ denotes the noise power spectral density, and ${W_{{\text{ISL}}}}$ represents the occupied bandwidth of ISL. 
	The above parameters remain constant in typical situations. However, the free-space path loss, i.e., $ {L_{{\rm{PL}}}}\left( {s_k^t,s_j^t} \right) $, varies according to the transmission distance $ d\left( {s_k^t,s_j^t} \right) $, and can be calculated as
	\begin{equation}\label{equ:FSPL}
		{L_{{\rm{PL}}}}\left( {s_k^t,s_j^t} \right) = {\left( {\frac{{\lambda \left( {{s_k},{s_j}} \right)}}{{4\pi d\left( {s_k^t,s_j^t} \right)}}} \right)^2},
	\end{equation}
	where $ \lambda \left( {{s_k},{s_j}} \right) $ is the wavelength.
	In addition, ${h^2}\left( {s_k^t,s_j^t} \right)$ denotes the small-scale fading variable, which will impact the link's PER and will be discussed in Section \ref{subsec:PER}.
	Then, the achievable rate of ISL $ l_f\left( {s_k^t,s_j^t} \right)  $ is given by
	\begin{equation}\label{equ:achievable_rate}
		C\left( {s_k^t,s_j^t} \right) = W_{\rm ISL}{\log _2}\left( {1 + {\gamma}\left( {s_k^t,s_j^t} \right)} \right).
	\end{equation}
	
	We use an SNR threshold $\gamma_{\rm{th}}$ to distinguish whether a link can perform communication. Specifically, the received SNR greater than $\gamma_{\rm{th}}$ is a prerequisite for establishing a link. Due to practical implementation, the maximum number of ISLs (denoted as $N_{{\rm{ISL}},\max}$) that can be established by each satellite during each connection attempt is generally limited\cite{9393372}. Thus, the connection relationships among satellites require adequate optimization. To this end, we introduce a \textit{satellite association} matrix ${\bf{X}} = \left\{ {\left. {{x_f}\left( {s_k^t,s_j^t} \right)} \right|f \in {\cal F},s_k^t,s_j^t \in {\cal S}} \right\}$, where ${{x_f}\left( {s_k^t,s_j^t} \right)}$ is a binary variable that represents the connectivity status of satellites. Specifically, ${x_f(s_k^t,s_j^t)} = 1$ signifies that the link $l_f(s_k^t,s_j^t) \in {\mathcal{L}}_{\rm{ISL}}$, and ${x_f(s_k^t,s_j^t)} = 0$ means the link remains idle.
	
	We also define a file \textit{download ratio} matrix $ {\bm{\rho}} = \left\{ {\left. {{\rho _f}\left( {s_k^t,s_j^t} \right),{\rho _f}\left( {g_m,s_j^t} \right)} \right|f \in {\cal F},s_j^t,s_k^t \in {\cal S},g_m \in {\cal G}} \right\} $, which represents the percentage of file $ f $ that is acquired by satellite $ s_j^t $ from other nodes (other satellites or the visible GS). Both ${\rho _f}\left( {s_k^t,s_j^t} \right)$ and ${\rho _f}\left( {g_m,s_j^t} \right)$ fall within the range of $ \left[0,1\right]$. When ${\rho _f}\left( {s_k^t,s_j^t} \right) $ or ${\rho _f}\left( {g_m,s_j^t} \right) $ equals 0 or 1, the partial downloading scenario simplifies to the binary downloading case, where all the requested contents are downloaded from a single node.
	Subsequently, the transmission delay of the file through ISL $ l_f\left( {s_k^t,s_j^t} \right) $  can be calculated as
	\begin{equation}\label{equ:delay}
		{t_f^{\rm{tran}}}\left( {s_k^t,s_j^t} \right) = \frac{{{\rho _f}\left( {s_k^t,s_j^t} \right){b_f}{R_p}}}{{{C}\left( {s_k^t,s_j^t} \right)}}.
	\end{equation}

	\subsubsection{G2S, G2A, S2A Links} We use the G2S link as an example to discuss the achievable rate, while a similar approach can be applied to determine the achievable rate for G2A and S2A links. Since the bandwidth of GSs is relatively limited compared with the laser frequency bands, we introduce a \textit{bandwidth allocation} variable ${\bm \omega}  = \left\{ {\left. {\omega \left( {{g_m},s_k^t} \right)} \right|  g_m \in {\mathcal{G}}, s_k^t \in {\mathcal{S}}} \right\} $, with $\omega\left( {{g_m},s_k^t} \right) \in \left[0,1\right]$ representing the proportion of the bandwidth occupied by satellite $s_k^t$. Therefore,
	the achievable rate of satellite $s_k^t$ when receiving file from GS $g_m$ is given by
	\begin{equation}\label{rate_g2s}
		{C}\left( {{g_m},s_k^t} \right) = \omega \left( {{g_m},s_k^t} \right)W_{\rm G2S}{\log _2}\left( {1 + \frac{{c\left( {{g_m},s_k^t} \right)}}{{\omega \left( {{g_m},s_k^t} \right)}}} \right)  ,
	\end{equation}
	where $c(g_m,s_k^t)= \frac{{{P_{\rm{T}}}\left( {{g_m}} \right){G_{\rm{T}}}\left( {{g_m}} \right){G_{\rm{R}}}\left( {{s_k^t}} \right){L_{{\rm{add-G2S}}}}{L_{{\rm{FSPL}}}}\left( {g_m,s_k^t} \right)}}{{{n_0}}W_{\rm G2S}}$ is a constant within each TS.
	
	\subsection{Packet Error Rate}\label{subsec:PER}
	Shadowed-Rician (SR) fading is a widely used distribution to characterize the small-scale fading of G2S links. Let ${\rm{SR}} \left(b_0, m, \Omega\right)  $ denote the SR fading. Here, $ \Omega $ is the average power of line-of-sight (LoS) component, $ 2b_0 $  is the average power of the multi-path component except the LoS one, and $ m $ is the Nakagami parameter.
	Then, based on the probability density function (PDF)  of SR fading power gain $ h^2 $ \cite{1198102}, involved in (\ref{equ:SNR}),  the PDF of received SNR $ \gamma $ under SR fading can be derived by
	\begin{equation}\label{key}
		\begin{split}
			p_{\rm{SR}}\left( \gamma  \right) & = {\left( {\frac{{2{b_0}m}}{{2{b_0}m + \Omega }}} \right)^m}\frac{1}{{2{b_0}\bar \gamma }}\exp \left( { - \frac{\gamma }{{2{b_0}\bar \gamma }}} \right)\\
			& \cdot {}_1{F_1}\left( {m,1,\frac{{\Omega \gamma }}{{2{b_0}\left( {2{b_0}m + \Omega } \right)\bar \gamma }}} \right),
		\end{split}
	\end{equation}
	where $ {}_1{F_1}\left(\cdot,\cdot,\cdot\right) $ is the confluent hypergeometric function.
	Here, $\bar \gamma$ is the average SNR and $\gamma = \bar \gamma h^2$ holds.
	For S2A links, Loo distribution is often used to model the small-scale fading \cite{54983}, denoted by ${\rm{Loo}}\left(\mu, d_0, b_0\right)$. Here, $\mu $ and $\sqrt{d_0}$ are the mean value and standard derivation due to shadowing, respectively.
	With the similar method, the PDF of received SNR $\gamma$ under Loo distribution can be expressed as
	\begin{equation}
		\begin{split}
			p_{\rm{Loo}}(\gamma )  & = \frac{\gamma}{b_0 {\bar \gamma}^2 \sqrt{2 \pi d_0}}  \int_0^{\infty} \frac{1}{z} \\
			& \cdot \exp \left(-\frac{(\ln z-\mu)^2}{2 d_0}-\frac{\left(\frac{\gamma^2}{\bar \gamma^2 } +z^2\right)}{2 b_0}\right) I_0\left(\frac{\gamma z}{{ \bar \gamma }  b_0} \right) {\rm{d}} z,
		\end{split}
	\end{equation}
	where $I_0\left(\cdot\right)$ denotes the modified Bessel function of zero order.
	
	Let $ f\left(\gamma \right) $ be the PER function for the AWGN channel with instantaneous SNR $ \gamma $. Assume that $ \int_0^{ + \infty } { f\left( \gamma  \right)} {\rm{d}}\gamma $ exists, and define $ c_0  = \int_0^{ + \infty } { f\left( \gamma  \right)}{\rm{d}}\gamma  $.
	The average PER $\bar \gamma$ over a SR fading channel with $ {\mathcal{SR}}\left(b_0, m, \Omega \right)  $ is upper bounded by
	\begin{equation}
		\begin{split}
			&	{\rm{PER}}_{\rm{SR}}\left( {\bar \gamma } \right) = \int_0^{ + \infty } {f\left( \gamma  \right){p_{\rm{SR}}}\left( \gamma  \right)} {\rm{d}}\gamma \\
			&	= {\left( {\frac{{2{b_0}m}}{{2{b_0}m + \Omega }}} \right)^m}\frac{1}{{2{b_0}\bar \gamma }}\int_0^{ + \infty } {\exp \left( { - \frac{\gamma }{{2{b_0}\bar \gamma }}} \right)} \\
			&	\cdot {}_1{F_1}\left( {m,1,\frac{{\Omega \gamma }}{{2{b_0}\left( {2{b_0}m + \Omega } \right)\bar \gamma }}} \right)f\left( \gamma  \right){\rm{d}}\gamma \\
			&	\leq {\left( {\frac{{2{b_0}m}}{{2{b_0}m + \Omega }}} \right)^m}\frac{1}{{2{b_0}\bar \gamma }}\int_0^{c_0 } {\exp \left( { - \frac{\gamma }{{2{b_0}\bar \gamma }}} \right)} \\
			&	\cdot {}_1{F_1}\left( {m,1,\frac{{\Omega \gamma }}{{2{b_0}\left( {2{b_0}m + \Omega } \right)\bar \gamma }}} \right){\rm{d}}\gamma, \\
		\end{split}
	\end{equation}
	where the inequality follows from the Lemma of integral inequalities in \cite{5703199}.
	Specifically, the PER function $ f\left(\gamma\right) $ of an uncoded packet system is given by $ f\left(\gamma\right) = 1-\left(1-b\left(\gamma\right)\right)^n $, where $ b\left(\gamma\right) $ is the BER for AWGN channel and $ n $ denotes the packet length in bits. Similarly, the upper bound of average PER over a Loo distribution fading is given by
	\begin{equation}
		\begin{split}
			&	{\rm{PER}}_{\rm{Loo}}\left( {\bar \gamma } \right) 	\leq \frac{1}{{{b_0}{{\bar \gamma }^2}\sqrt {2\pi {d_0}} }}\int_0^{{c_0}} {\gamma \int_0^\infty  {\frac{1}{z}} } \\
			&	\cdot \exp \left( { - \frac{{{{(\ln z - \mu )}^2}}}{{2{d_0}}} - \frac{{\left( {\frac{{{\gamma ^2}}}{{{{\bar \gamma }^2}}} + {z^2}} \right)}}{{2{b_0}}}} \right){I_0}\left( {\frac{{\gamma z}}{{\bar \gamma {b_0}}}} \right){\rm{d}}z{\rm{d}}\gamma , \\
		\end{split}
	\end{equation}
	
	Fig. \ref{fig:SR_PER} plots the PER versus average SNR within certain range under three scenarios. That is, infrequent light shadowing (ILS), frequent heavy shadowing (FHS), and average shadowing (AS). The $ {\rm{SR}}\left(b_0, m, \Omega\right) $ and ${\rm{Loo}}\left(\mu, d_0, b_0\right)$ are set according to \cite{1198102} and packet length is 1023 bits. We can find that for the most common scenario (like ILS and AS), when the average SNR of SR fading channel is over 16 dB, PER can drop below 1\%, while the SNR threshold of Loo distribution is about 7 dB.
	This complies with the 3GPP requirements for streaming services \cite{3gpp2022}.
	The SNR of each type of link in IFC-oriented SAGIN can easily reach the above SNR threshold such that the packet loss can be neglected in the following.

	\begin{figure}[t]
		\centering
		\includegraphics[width = 0.4\textwidth]{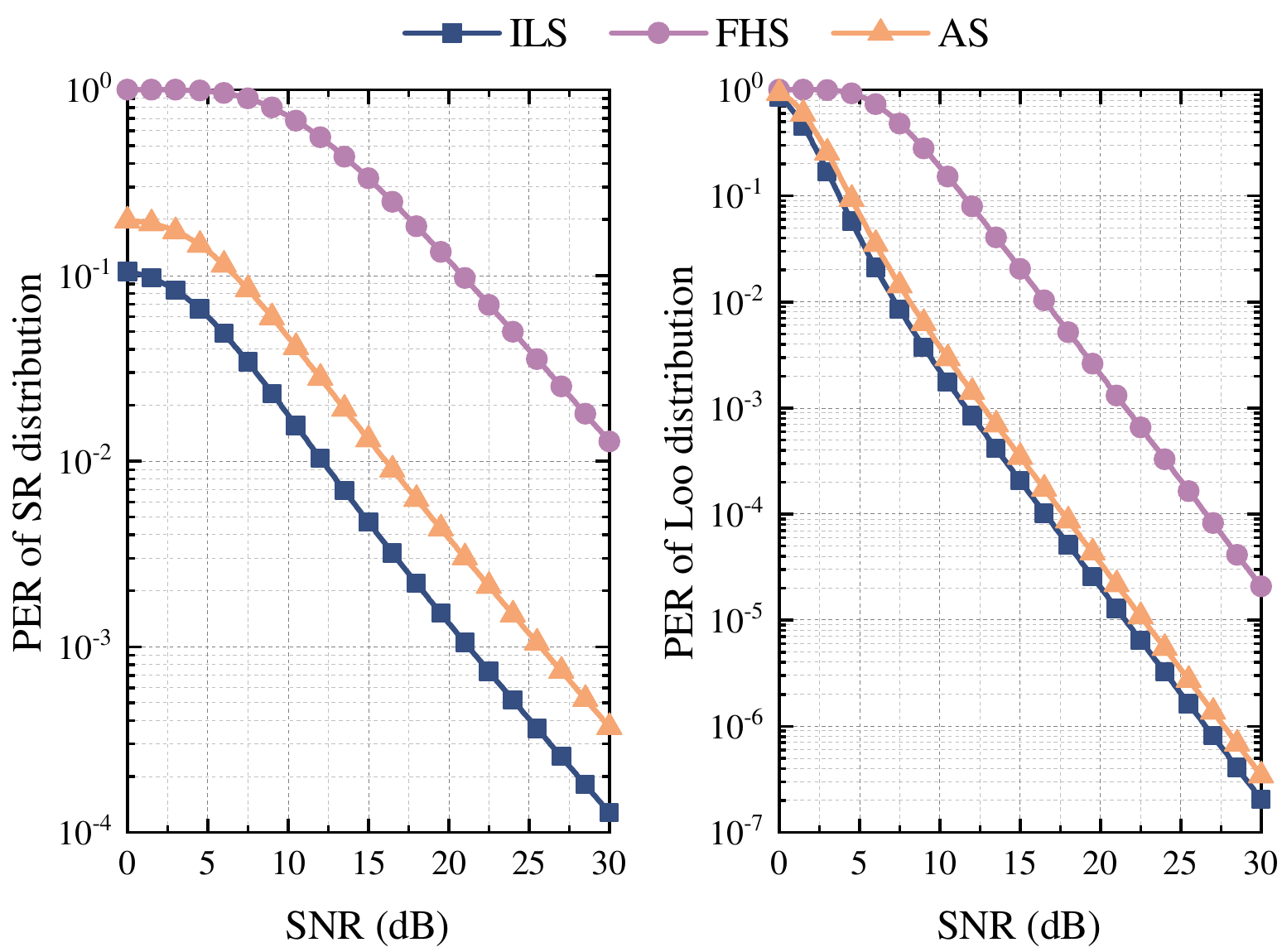}
		\caption{PER vs average SNR under SR (left) and Loo (right) distributions.}\label{fig:SR_PER}  
	\end{figure}

	\section{Delivery Scheme of the Cached Files}\label{sec:cached}
	In this section, we first introduce the cooperative transmission process of the cached files. Then, we formulate a delay minimization problem and propose an exact penalty-based algorithm to solve it efficiently.

	\subsection{Delivery Process of the Cached Files}
	Fig. \ref{fig:parallel_delay_cache} illustrates the cooperative transmission process of the cached files when aircraft connects to the nearest satellite. 
	After the aircraft $ a_i^t \in {\mathcal{A}} $ broadcasts its content request, the aircraft will be aware of which nodes have cached the corresponding files. When $ a_i^t$ can establish a link with the gateway $ g_m \in {\mathcal{G}} $, the requested content $ f $ can be directly transmitted via the G2A link $l_f \left(g_m,a_i^t\right) \in {\mathcal{L}}_{\rm{G2A}}$.
	However, if aircraft $a_i^t$ is unable to establish any G2A links (i.e.  $\gamma_f \left( {g_m,a_i^t} \right) < {\gamma _{{\rm{th}}}},\forall g_m \in {\mathcal G} $), it will instead connect to the nearest satellite $s_j^t \in {\mathcal S}$. Generally, satellite $s_j^t$ only possesses a portion of the required content. The non-cached files must be acquired from other nodes capable of establishing links with satellite $s_j^t$, termed as \textit{Phase 1}. This situation, as shown in Fig. \ref{fig:cache_case}, can be categorized into two scenarios:
	\begin{enumerate}[label=(\arabic*)]
		\item When satellite $s_j^t$ cannot connect to any GS, it establishes connections with neighboring satellites $s_k^t$ to obtain files through ISLs denoted as $l_f\left(s_k^t,s_j^t\right) \in {\mathcal{L}}_{\rm{ISL}}$.

		\item Satellite $s_j^t$ simultaneously downloads files from adjacent satellites $s_k^t$ via ISLs $l_f\left(s_k^t,s_j^t\right)$ and from the nearest satellite GS $g_m$ through the G2S link $l_f\left(g_m,s_j^t\right)$. 
	\end{enumerate}
	Here, the first case can be regarded as the special case of the second one.
	After $s_j^t$ obtains the whole required file, it will aggregate the content and send it to $a_i^t$, termed as \textit{Phase 2}.
	Particularly, when satellite $s_j^t$ already possesses all the content requested by aircraft $ a_i^t $, the file can be directly transmitted using the S2A link $l_f\left(s_j^t,a_i^t\right) \in {\mathcal{L}}_{\rm{S2A}}$, like the way that the aircraft downloads file from the nearest aircraft gateway. In this case, transmission occurs exclusively during Phase 2 without Phase 1, as shown in Fig. \ref{fig:parallel_delay_cache}.

	\begin{figure}[t]
		\centering
		\includegraphics[width = 0.3\textwidth]{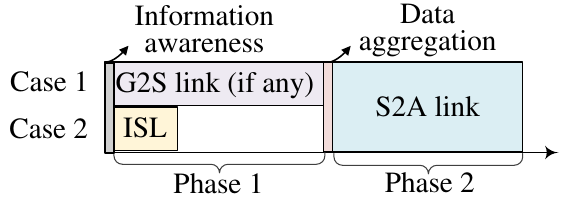}
		\caption{Timeline of cooperative transmission of cached files when aircraft connects to the nearest satellite. \label{fig:parallel_delay_cache}}
	\end{figure}

	\begin{figure}[t]
		\centering
		\includegraphics[width = 0.28\textwidth]{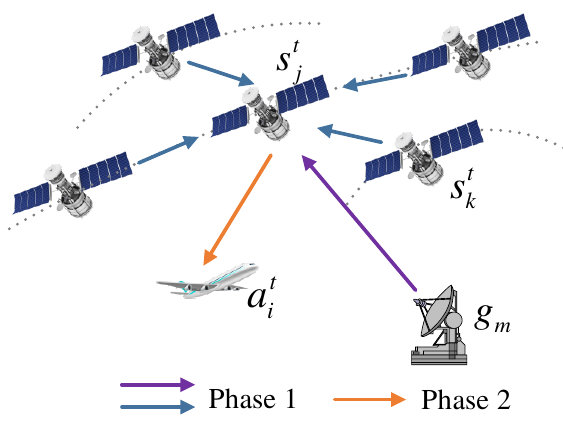}
		\caption{Delivery process of the cached files when aircraft download the files from satellite networks. \label{fig:cache_case}}
	\end{figure}
	
	The stages of information awareness and data aggregation are of negligible duration. Subsequently, we will discuss the delay associated with Phase 1 and 2 when the aircraft leverage the satellite network for content delivery.
	Under parallel transmission, the delay through all ISLs during Phase 1 is given by
	\begin{equation}\label{equ:delay_ISL}
		\begin{split}
			& {t_{f,{\rm{ISL}}}^{\rm{cached}}}\left( {s_j^t} \right) \\
			& = \mathop {\max }\limits_{s_k^t} \left\{ {{x_f}\left( {s_k^t,s_j^t} \right) \cdot \left( {t_f^{{\rm{tran}}}\left( {s_k^t,s_j^t} \right) + t_f^{{\rm{prop}}}\left( {s_k^t,s_j^t} \right)} \right)} \right\},\\
			& \qquad \qquad  \qquad \qquad \qquad \qquad \qquad \qquad  \qquad \qquad \forall s_j^t \in {\cal S}.
		\end{split}
	\end{equation}
	Here, $ {t_f^{\rm{prop}}}\left( {s_k^t,s_j^t} \right) = \frac{{{d_f}\left( {s_k^t,s_j^t} \right)}}{c} $ represents the propagation delay, where $ c $ is the light velocity.\footnote{It is reasonable to neglect the propagation delay since it is small (i.e., a few milliseconds) compared with the transmission delay, which is of several seconds.}
	The transmission and propagation delays of the G2S links can be calculated similarly to those of the ISLs, which are omitted here. Therefore, when aircraft $a_i^t$ connects to satellite $s_j^t$ to retrieve cached file $f$, the total delay is given by
	\begin{equation}\label{equ:delay_SN_cache}
		\begin{split}
			& {t_{f,{\rm{SN}}}^{\rm{cached}}}\left( {s_j^t,a_i^t} \right) \\
			& = \underbrace {\max \left\{ {{t_{f,{\rm{ISL}}}^{\rm{cached}}}\left( {s_j^t} \right),{t_{f,{\rm{G2S}}}}\left( {g_m,s_j^t} \right)} \right\}}_{{\rm{Phase\,1}}} + \underbrace {{t_{f,{\rm{S2A}}}}\left( {s_j^t,a_i^t} \right)}_{{\rm{Phase\,2}}}, \\
			& \qquad \qquad  \forall s_j^t,g_m:l_f\left( {s_j^t,a_i^t} \right),l_f\left( {g_m,s_j^t} \right) \in {\mathcal L},
		\end{split}
	\end{equation}
	where $ {t_{f,{\rm{G2S}}}}\left( {g_m,s_j^t} \right) = t_f^{{\rm{tran}}}\left( {g_m,s_j^t} \right) + t_f^{{\rm{prop}}}\left( {g_m,s_j^t} \right) $ and $ {t_{f,{\rm{S2A}}}}\left( {s_j^t,a_i^t} \right) = t_f^{{\rm{tran}}}\left( {s_j^t,a_i^t} \right) + t_f^{{\rm{prop}}}\left( {s_j^t,a_i^t} \right) $ represent the sum (transmission and propagation) delay by G2S link and S2A link, respectively.
	Particularly, when satellite $s_j^t$ cannot set up G2S links with any GSs, $C\left( {g_m,s_j^t} \right) = 0$.
	Hence, the total delay of downloading file $f$ by aircraft $a_i^t = {\rm{sn}}_f$ is given by
	\begin{equation}\label{equ:delay_f}
		\resizebox{1\hsize}{!}{$	\begin{split}
				& {t_f^{\rm{cached}}}\left( {a_i^t} \right) =\\
				& \left\{ \begin{array}{l}
					{t_f}\left( {g_m,a_i^t} \right),{\rm{if }} \ l_f\left( {g_m,a_i^t} \right) \in {{\mathcal L}_{{\rm{G2A}}}},\\
					{t_{f,{\rm{SN}}}^{\rm{cached}}}\left( {s_j^t,a_i^t} \right), {\rm{if }} \ l_f\left( {g_m,a_i^t} \right) \notin {{\mathcal L}_{{\rm{G2A}}}} \ {\rm{and}} \ l_f\left( {s_j^t,a_i^t} \right) \in {{\mathcal L}_{{\rm{S2A}}}}.\\
				\end{array} \right.
			\end{split} $}
	\end{equation}
	
	
	\subsection{Problem Formulation}
	This paper aims to minimize the overall downloading delay for the airborne Internet access issue by jointly optimizing the satellite association $\{\bf X\}$ and the file download ratio $\{\bm \rho\}$.
	Note that minimizing the total delay $\{{t_f^{\rm cached}}\left( {a_i^t} \right)\}$ is equivalent to minimizing the sum delay of Phase 1 since all other parts are constant. 
	Hence, this optimization problem can be formulated as
	\begin{subequations}
		\begin{equation}\label{p1}
			{\mathcal{P}}1:	\mathop {\min }\limits_{{\bf{X}},{\bm{\rho} }} \ {\sum\limits_f {\max \left\{ {{t_{f,{\rm{ISL}}}^{\rm{cached}}}\left( {s_j^t} \right),{t_{f,{\rm{G2S}}}}\left( {g_m,s_j^t} \right)} \right\}} } 
		\end{equation}
		\begin{equation}\label{p1_maxISL}
			{\rm{s.t.}}  \quad   \sum\limits_{s_k^t \in {{\mathcal S}}} {{x_f}\left( {s_k^t,s_j^t} \right)}  \le {N_{{\rm{ISL}},\max }}, \; \forall s_j^t\in\mathcal{S},
		\end{equation}
		\begin{equation} \label{p1_rho}
			{\rho _f}\left( {s_k^t,s_j^t} \right) \le {x_f}\left( {s_k^t,s_j^t} \right), \; \forall s_j^t, s_k^t  \in \mathcal{S},
		\end{equation}
		\begin{equation}\label{p1_sumrho}
			\sum\limits_{s_k^t \in {{\mathcal S}}} {{\rho _f}\left( {s_k^t,s_j^t} \right)}  + {\rho _f}\left( {g_m,s_j^t} \right)  = 1, \; \forall s_j^t\in\mathcal{S},\forall g_m\in\mathcal{G},
		\end{equation}
		\begin{equation}\label{cst:x_binary}
			{x_f}\left( {s_k^t,s_j^t} \right) \in \left\{ {0,1} \right\},\forall s_j^t, s_k^t  \in\mathcal{S},
		\end{equation}
		\begin{equation}\label{cst:rho_continuous}
			{\rho_f}\left( {s_k^t,s_j^t} \right) \in \left[ {0,1} \right],\forall s_j^t, s_k^t  \in \mathcal{S},
		\end{equation}
	\end{subequations}
	where \eqref{p1_maxISL} ensures that the total number of established ISLs for each satellite within a given time slot does not exceed the maximum allowable value $ N_{{\rm{ISL}},\max } $, 
	\eqref{p1_rho} means that the downloading ratios from adjacent satellites are nonzero only when the ISLs are set up, and \eqref{p1_sumrho} dictates that the sum of the download ratios of data obtained from adjacent satellites and its GS equals 1. In addition, \eqref{cst:x_binary} and \eqref{cst:rho_continuous} limit the upper and lower bounds of $\bf X$ and $\bm \rho$.
	
	\begin{remark}
		Thanks to the considerable capacity provided by the laser ISLs, each file request can be readily fulfilled within the considered time duration $\tau$. Hence, it is reasonable to assume that the content delivery is independent for each time slot. For simplicity, we focus on a typical time slot and omit the superscript $t$ in $ {x_f}\left( {s_k^t,s_j^t} \right)$ and $ s_k^t $ herein-below. 	
	\end{remark}
	
	\subsection{Proposed Exact Penalty Method}
	
	$\mathcal{P}1$ is a mixed integer nonlinear programming (MINLP) problem, which is non-trivial to solve directly. To facilitate solving it, we first provide the following lemma to show the properties of the optimal solution of $\mathcal{P}1$.
	\begin{lemma}\label{lemma1}
		The optimal solution of $\mathcal{P}1$ should satisfy the following conditions:
		\begin{itemize}
			\item For each requested file, the transmission delay through each established ISL (and G2S, if visible) should be equal.
			\item For each satellite that aggregates the requested file (i.e., $s_j$), the number of ISLs it established equals to  ${N_{{\rm{ISL}},\max}}$ or the maximum number of visible satellites that have cached the file.
		\end{itemize}
	\end{lemma}
	\begin{proof}
		According to \eqref{equ:delay_ISL}, the objective function in $\mathcal{P}1$ takes the `minimize-the-maximum' form. As such, the optimal transmission delay for each established link should be consistent. In other words, if there exists any link whose delay exceeds others', we can always increase the download ratio of other links to balance the delay. Secondly, each satellite tends to establish as many ISLs as possible to fully use the high-capacity laser ISLs to minimize the transmission delay. 
	\end{proof}
	
	Based on Lemma \ref{lemma1}, we have for any $s_j$ and $s_k$ that satisfy ${x^*_f}\left( {s_k,s_j} \right)=1$, the optimal transmission delay of each ISL/G2S for $s_j$ can be calculated as
	\begin{equation}\label{tranformed_delay}
		\begin{aligned}
			{t^*_{f,{\rm{ISL}}}}\left( {s_j}\right)=
			\frac{{{\rho _f}\left( {s_k,s_j} \right){b_f}{R_p}}}{{{C}\left( {s_k,s_j} \right)}}=
			\frac{{{\rho _f}\left( {g_m,s_j} \right){b_f}{R_p}}}{{{C}\left( {g_m,s_j} \right)}}=\\
			\frac{{{b_f}{R_p}}}{\sum\limits_{s_k \in {{\mathcal S}}}{x_f}\left( {s_k,s_j} \right){{C}\left( {s_k,s_j} \right)}+C(g_m,s_j)}.
		\end{aligned}
	\end{equation} 
	
	According to \eqref{tranformed_delay}, the continuous variable $\rho_f$ can be eliminated, and $\mathcal{P}1$ can be equivalently transformed into the following binary integer programming (BIP) problem:
	\begin{subequations}
		\begin{align}
			{\mathcal{P}}2:	\mathop {\min }\limits_{\bf{X}}& \sum\limits_{s_j \in {{\mathcal S}}}
			\frac{{{b_f}{R_p}}}{\sum\limits_{s_k \in {{\mathcal S}}}{x_f}\left( {s_k,s_j} \right){{C}\left( {s_k,s_j} \right)}+C(g_m,s_j)} \\
			\label{p2_maxISL}{\rm s.t.} \quad   & \sum\limits_{s_k \in {{\mathcal S}}} {{x_f}\left( {s_k,s_j} \right)}  \le {N_{{\rm{ISL}},\max }}, \; \forall s_j \in {\mathcal S}, \\
			& {x_f}\left( {s_k,s_j} \right) \in \{0,1\}, \forall s_j, s_k \in {\mathcal S}. \label{p2_bi}
		\end{align}
	\end{subequations}

	\begin{remark}\label{remark2}(ISL selection for single or sparse file request)
		From ${\mathcal{P}}2$, it can be observed that if there is only one satellite sending file request, we can simply select ${N_{{\rm{ISL}},\max }}$ adjacent satellites with the highest capacity to minimize the total delay. This conclusion is also true when the file requests are sparse. However, when the number of requests is large, link $l_f(s_j,s_k)$ with the highest capacity for $s_j$ may not be a desired link for $s_k$. Consequently, directly selecting the ISLs with the highest capacity may conflict with constraint \eqref{p2_maxISL}. This inspires us to devise an efficient satellite association scheme.
	\end{remark}
	
	The transformed $\mathcal{P}2$ is still challenging due to the following reasons: First, the relaxation and rounding-based method is not suitable since ${x_f}\left( {s_k,s_j} \right) $ will be the same as ${\rho_f}\left( {s_k,s_j} \right)$ after relaxation and thus loses its physical meaning. Moreover, the exhaustive search-based methods, such as branch and bound, are computationally prohibitive since the problem is generally large-scale for a mega satellite constellation. 
	Therefore, to tackle this BIP problem, we propose an exact penalty method (EPM) to solve it with low complexity. Specifically, we first provide the following theorem to replace the binary constraint \eqref{p2_bi} with several equality constraints: 
	\begin{theorem}\label{theorem_1}
		Define $\mathbf{u}\in\mathbb{R}^{m\times n}$, $\mathbf{v}\in\mathbb{R}^{m\times n}$ and  $\Theta \triangleq \left\lbrace (\mathbf{u},\mathbf{v})|\langle \mathbf{u},\mathbf{v}\rangle=mn,-1\preceq \mathbf{u} \preceq 1, \Vert \mathbf{v}\Vert_2^2\leq mn \right\rbrace $. If $(\mathbf{u},\mathbf{v}) \in \Theta$, then $\mathbf{u}\in \{-1,1\}^{m\times n}$,  $\mathbf{v}\in \{-1,1\}^{m\times n}$, and $\mathbf{u}=\mathbf{v}$.
	\end{theorem}
	
	\begin{proof}
		Please see Appendix A.
	\end{proof}
	
	Based on Theorem \ref{theorem_1}, by introducing ${\pi_f}\left( {s_k,s_j} \right)=2{x_f}\left( {s_k,s_j} \right)-1$, constraint \eqref{p2_bi} is equivalent to three constraints: $\langle \boldsymbol{\pi},\mathbf{v}\rangle=S^2,-1\preceq \boldsymbol{\pi} \preceq 1, \Vert \mathbf{v}\Vert_2^2\leq S^2$. Then, $\mathcal{P}2$ is equivalent to the following optimization problem:
	\begin{subequations}
		\begin{equation}
			\resizebox{1\hsize}{!}{$	{\mathcal{P}}3:	 \mathop {\min } \limits_{\boldsymbol{\pi},\mathbf{v}} \; F(\mathbf{\boldsymbol{\pi},\bf{v}})  \triangleq  \sum\limits_{s_j \in {{\mathcal S}}}
				\frac{{{b_f}{R_p}}}{\sum\limits_{s_k \in {{\mathcal S}}}\frac{{\pi_f}\left( {s_k,s_j} \right)+1}{2}{{C}\left( {s_k,s_j} \right)}+C(g_m,s_j)}  $}
		\end{equation} 
		\begin{equation}
			{\rm s.t.} \quad   \sum\limits_{s_k \in {{\mathcal S}}} \frac{{\pi_f}\left( {s_k,s_j} \right)+1}{2}  \le {N_{{\rm{ISL}},\max }}, \; \forall s_j \in {\mathcal S},
		\end{equation}
		\begin{equation}\label{p3_bi}
			-1\preceq \boldsymbol{\pi} \preceq 1, \quad  \Vert \mathbf{v}\Vert_2^2\leq S^2,
		\end{equation}
		\begin{equation}\label{p3_pel}
			\langle \boldsymbol{\pi},\mathbf{v}\rangle=S^2.
		\end{equation}
	\end{subequations}

	To tackle the equality constraint \eqref{p3_pel} in $\mathcal{P}3$, we add a penalty term to $F(\mathbf{\boldsymbol{\pi},\bf{v}})$, yielding the following optimization problem:
	\begin{subequations}
		\begin{align}
			{\mathcal{P}}4:	\mathop {\min }\limits_{\boldsymbol{\pi},\mathbf{v}}& \; \mathcal{J}_\rho(\boldsymbol{\pi},\mathbf{v}) \triangleq  F(\mathbf{\boldsymbol{\pi},\bf{v}}) +\epsilon\left( S^2-\langle \boldsymbol{\pi},\mathbf{v}\rangle\right)  \\
			{\rm s.t.} \quad  & \sum\limits_{s_k \in {{\mathcal S}}} \frac{{\pi_f}\left( {s_k,s_j} \right)+1}{2}  \le {N_{{\rm{ISL}},\max }}, \; \forall s_j \in {\mathcal S}, \\
			&-1\preceq \boldsymbol{\pi} \preceq 1,  \quad   \Vert \mathbf{v}\Vert_2^2\leq S^2, \label{p4_bi}
		\end{align}
	\end{subequations}
	where $\epsilon$ is the penalty parameter, which is gradually increased to enforce the equality of \eqref{p3_pel}. Specifically, the penalty parameter is updated by
	\begin{equation}\label{equ:update}
		\epsilon^{(l+1)}=\epsilon^{(l)}\times \Delta,
	\end{equation}
	where $\Delta$ is a constant and $l$ is the number of iteration. Then, in the $l$-th iteration, we optimize $\boldsymbol{\pi}$ and $\mathbf{v}$ in an iterative manner. The detailed procedure for solving $\mathcal{P}4$ is introduced as follows.
	
	\subsubsection{Initialization}
	In the first iteration, i.e., $l=0$, $\epsilon$ is initialized as a relatively small number approaching zero. Besides, we initialize $\mathbf{v}^{(0)} = \mathbf{0}$ to find a reasonable local-optimal point.

	\subsubsection{Optimize $\bm{ \pi}$ for given $\mathbf{v}$} 
	For each fixed $\epsilon$ and $\bf v$, the $\boldsymbol{\pi}$-subproblem can be written as
	\begin{subequations}
		\begin{align}
			{\mathcal{P}}{\rm{4\text{-}1}}:	\mathop {\min }\limits_{\boldsymbol{\pi}}& \; \mathcal{J}_\rho(\boldsymbol{\pi}^{(l+1)},\mathbf{v}^{(l)})   \\
			{\rm s.t.} \quad  & \sum\limits_{s_k \in {{\mathcal S}}} \frac{{\pi_f}\left( {s_k,s_j} \right)+1}{2}  \le {N_{{\rm{ISL}},\max }}, \; \forall s_j \in {\mathcal S}, \\
			&-1\preceq \boldsymbol{\pi} \preceq 1. \label{p41_bi}
		\end{align}
	\end{subequations}
	${\mathcal{P}}{\rm{4\text{-}1}}$ is a convex optimization problem, whose optimal solution can be obtained via standard optimization toolbox such as CVX\cite{cvx}.
	
	\subsubsection{Optimize $\bf{v}$ for given ${\bm \pi}$}  
	The $\mathbf{v}$-subproblem can be written as
	\begin{subequations}
		\begin{align}
			{\mathcal{P}}{\rm{4\text{-}2}}:	\mathop {\min }\limits_{\mathbf{v}}& \; \langle -\boldsymbol{\pi}^{(l+1)},\mathbf{v}^{(l+1)}\rangle   \\
			{\rm s.t.} \quad  & \; \Vert \mathbf{v}\Vert_2^2\leq S^2.
		\end{align}
	\end{subequations}
	It can be observed from ${\mathcal{P}}{\rm{4\text{-}2}}$ that unless $\boldsymbol{\pi}\neq\mathbf{0}$, the optimal solution will be achieved in the boundary with $\Vert \mathbf{v}\Vert_2^2=S^2$. Hence, the optimal solution to the $\mathbf{v}$-subproblem can be expressed in a closed form as:
	\begin{equation}\label{opt_v}
		\mathbf{v}^{(l+1)}=\begin{cases}
			S\cdot \boldsymbol{\pi}^{(l+1)}/\Vert\boldsymbol{\pi}^{(l+1)}\Vert_2^2, \;\;\;\text{if} \; \boldsymbol{\pi}^{(l+1)}\neq\mathbf{0},\\
			\text{any} 	\;\mathbf{v} \; \text{with}	\Vert \mathbf{v}\Vert_2^2\leq S^2, \; \text{otherwise}. 
		\end{cases}
	\end{equation}
	
	\begin{algorithm}[!t]
		\caption{Proposed EPM for solving $\mathcal{P}$4}
		\begin{algorithmic}[1]\label{alg1}
			\State {Initialize iteration number $l=0$, $\boldsymbol{\pi}^{(0)}=\mathbf{v}^{(0)}=\mathbf{0}$, $\epsilon^{(0)}>0$, $\Delta>0$, convergence threshold $\xi$.}
			\State \quad {\bf repeat}   
			\State \quad\quad  Update $\boldsymbol{\pi}^{(l+1)}$ by solving ${\mathcal{P}}{\rm{4\text{-}1}}$.
			\State \quad\quad  Update $\mathbf{v}^{(l+1)}$ by \eqref{opt_v}.
			\State \quad\quad  Update the penalty factor.
			\State \quad\quad $l=l+1$.
			\State \quad {\bf until} $|S^2-\langle \boldsymbol{\pi},\mathbf{v}\rangle|\leq \xi$.
			\State Transform $\boldsymbol{\pi}$ into ${x_f}\left( {s_k^t,s_j^t} \right)$.
		\end{algorithmic}
	\end{algorithm}
	
	The detailed procedure for solving $\mathcal{P}4$ is summarized in Algorithm \ref{alg1}.
	
	\subsection{Complexity Analysis}
	At each iteration of Algorithm \ref{alg1}, the computational complexity involved in solving the convex problems ${\mathcal{P}}{\rm{4\text{-}1}}$ remains polynomial with respect to the number of variables and constraints. We consider the worst case, wherein all the aircraft have content delivery requires and their associated satellite can set up $ N_{\rm{ISL},\max}$ ISLs with neighboring satellites. Then, ${\mathcal{P}}{\rm{4\text{-}1}}$ is an optimization problem with $A N_{\rm{ISL},\max}$ real-valued variables, a linear objective, and $A + A N_{\rm{ISL},\max}$ linear constraints. Thus, the complexity required to solve ${\mathcal{P}}{\rm{4\text{-}1}}$ is upper bounded by $\mathcal{O}\left( {\left( {A + 2A{N_{{\text{ISL}},\max }}} \right){{\left( {A{N_{{\text{ISL}},\max }}} \right)}^2}\sqrt {A + A{N_{{\text{ISL}},\max }}} } \right)$. Eq. (\ref{opt_v}) is a closed-form expression with complexity $\mathcal{O}\left( {A{N_{{\text{ISL}},\max }}} \right)$ and the complexity of updating the penalty factor in (\ref{equ:update})  is $\mathcal{O}\left(1\right)$.
	
	The computation complexity of exhaustive method exhibits an exponential growth relationship with the number of binary variables. That is, the computational complexity of only solving the binary optimization problem is upper bounded by $\mathcal{O}\left(2^{A{N_{{\text{ISL}},\max }}}\right)$, which is much larger than that of our proposed schemes.
	
	\section{Delivery Scheme of the Non-Cached Files}\label{sec:noncached}
	In this section, we introduce the delivery schemes of the non-cache files, which are not cached on non-terrestrial platforms and can only be acquired from GSs. Compared to downloading cached files, the bandwidth allocation of GSs should also be considered in this case.
	We propose an efficient algorithm based on alternating optimization to solve the corresponding delay minimization problem.
	
	\begin{figure}[h]
		\centering
		\includegraphics[width = 0.3\textwidth]{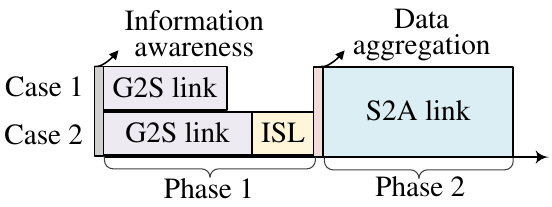}
		\caption{Timeline of cooperative transmission of non-cached files when aircraft connects to the nearest satellite. \label{fig:parallel_delay_uncache}}
	\end{figure}
	
	\subsection{Delivery Process of Non-Cached Files} 
	Similar to section \ref{sec:cached}, we omit the simple case that the aircraft $a_i^t$ can directly download files via the G2A link. We focus that the aircraft obtains the file from satellite GSs via its nearest satellite $s_j^t$, while other satellites act as relay nodes. 
	Then, the cooperative transmission scheme, as shown in Fig. \ref{fig:parallel_delay_uncache} and Fig. \ref{fig:noncache_case}, can be categorized into two cases:
	\begin{enumerate}[label=(\arabic*)]
		\item When satellite $s_j^t$ is visible to a satellite GS, $s_j^t$ retrieves all the content from the nearest satellite GS $g_m$ directly.

		\item When satellite $s_j^t$ cannot set up a G2S link with any satellite GS, $s_j^t$ connects to its multiple neighboring satellites $s_k^t$ which are visible to satellite GSs.\footnote{We consider at most two-hop transmission in this paper, while the extension to multi-hop relay is left as our future work.} Each $s_k^t$ downloads the file from its GS ${{g_m}}$ and then transmits it to $s_j$.
		
	\end{enumerate}
	\begin{figure}[t]
		\centering
		\includegraphics[width = 0.33\textwidth]{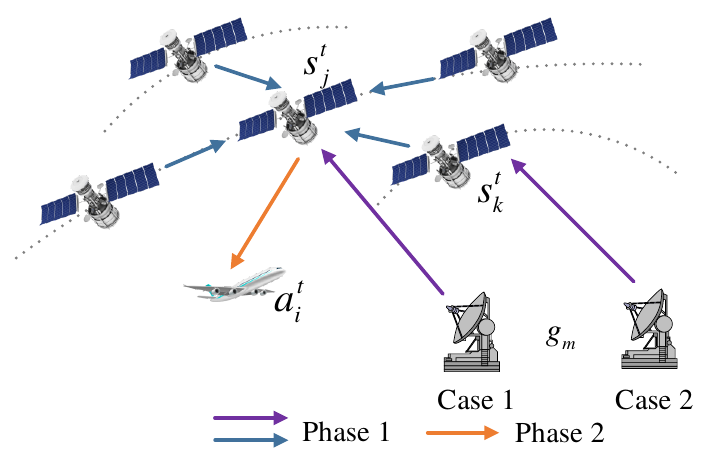}
		\caption{Delivery process of the non-cached files when aircraft download the files from satellite networks. \label{fig:noncache_case}}
	\end{figure}
	
	Here, the total delay of case 1 is a constant, which can be calculated as
	\begin{equation}
		\begin{split}
			& t_{f,{\rm{SN}}}^{{\rm{non}},{\rm{case1}}}\left( {s_j^t,a_i^t} \right) \\
			& = \underbrace {t_f^{{\rm{tran}}}\left( {{g_m},s_j^t} \right) + t_f^{{\rm{prop}}}\left( {{g_m},s_j^t} \right)}_{{\rm{Phase\,1}}} + \underbrace {{t_{f,{\rm{S2A}}}}\left( {s_j^t,a_i^t} \right)}_{{\rm{Phase\,2}}}.
		\end{split}
	\end{equation}
	
	We focus on the delay characterization of case 2 in the following. Similar to Section \ref{sec:cached}, the optimization variables of non-cached files also involve the satellite association strategy between $s_j^t$ and its neighboring satellites $s_k^t$, as well as the download ratio of satellite $s_j^t$ from GS $g_m$ via $s_k^t$. To differentiate between cached and non-cached delivery schemes, we use ${{\tilde{\bf X} }}$ and ${\tilde{\bm{\rho}}}$ to represent the corresponding satellite association and download ratio variables for transmitting non-cached files, where $\tilde{{\bf{X}} }= \left\{ {\left. {{{\tilde x}_f}\left( {s_k^t,s_j^t} \right)} \right|f \in {\cal F},s_j^t,s_k^t \in {\cal S}} \right\}$ and 
	$ \tilde{\bm \rho} = \left\{ {\left. {{{\tilde \rho} _f}\left( {s_k^t,s_j^t} \right)} \right|f \in {\cal F},s_k^t,s_j^t \in {\cal S}} \right\} $. Moreover, since the frequency resource of terrestrial networks is scarce, the bandwidth allocation $\bm \omega = \left\{ {\left. {{{\omega}}\left( {g_m, s_k^t} \right)} \right| g_m \in {\mathcal{G}}, s_k^t \in {\cal S}} \right\} $ requires proper optimization.

	When satellite $s_j^t$ downloads file $f$ from GS $g_m$ via neighboring satellite $s_k^t$, the two-hop delay can be expressed as
	\begin{equation}\label{equ:uncache_delay_G2S}
		\resizebox{1\hsize}{!}{$ t_f^{{\rm{relay}}}\left( {{g_m},s_k^t,s_j^t} \right) = t_f^{{\rm{tran}}}\left( {{g_m},s_k^t} \right) + t_f^{{\rm{tran}}}\left( {s_k^t,s_j^t} \right) + t_f^{{\rm{prop}}}\left( {{g_m},s_j^t} \right), $}
	\end{equation} 
	where $t_f^{{\rm{tran}}}\left( {{g_m},s_k^t} \right)$ and $t_f^{{\rm{tran}}}\left( {s_k^t,s_j^t} \right)$ denote the transmission delay of G2S and ISL, respectively, and $t_f^{{\rm{prop}}}\left( {{g_m},s_j^t} \right)$ denotes the sum of propagation delay.
	Note that the two terms of transmission delay in (\ref{equ:uncache_delay_G2S}) are related to the ratio ${\tilde \rho _f}\left( {s_k^t,s_j^t} \right)$.
	
	Then, the delay of $s_j^t$ downloading file $f$ with the assistance of all other satellites via ISLs, termed as \textit{Phase 1}, can be expressed as
	\begin{equation}
		t_{f, \rm ISL}^{{\rm{non, case2}}}\left( {s_j^t} \right) = \mathop {\max }\limits_{s_k^t} \left\{ {{\tilde x_f}\left( {s_k^t,s_j^t} \right) \cdot t_f^{{\rm{relay}}}\left( {{g_m},s_k^t,s_j^t} \right)} \right\}.
	\end{equation}
	
	Therefore, we can obtain the total delay when aircraft $a_i^t$ connects to satellite $s_j^t$ to obtain non-cached file $f$ from $g_m$, given as follows.
	\begin{equation}
		t_{f,{\rm{SN}}}^{{\rm{non, case2}}}\left( {s_j^t,a_i^t} \right) = \underbrace {t_{f, \rm ISL}^{{\rm{non, case2}}}\left( {s_j^t} \right)}_{{\rm{Phase\,1}}} + \underbrace {{t_{f,{\rm{S2A}}}}\left( {s_j^t,a_i^t} \right)}_{{\rm{Phase\,2}}},
	\end{equation}
	where ${{t_{f,{\rm{S2A}}}}\left( {s_j^t,a_i^t} \right)}$ in Phase 2 is determined in a similar way as (\ref{equ:delay_SN_cache}). 
	
	Finally, the total delay for downloading non-cached file $f$ by aircraft $a_i^t = {\rm{sn}}_f$ can be expressed as follows.
	\begin{equation}
		\begin{split}
			& t_f^{{\rm{non-cached}}}\left( {a_i^t} \right) \\
			& = \left\{ {\begin{array}{*{20}{l}}
					{{t_f}\left( {{g_m},a_i^t} \right),{\rm{if}}\;{l_f}\left( {{g_m},a_i^t} \right) \in {{\cal L}_{{\rm{G2A}}}},}\\
					{t_{f,{\rm{SN}}}^{{\rm{non, case1}}},{\rm{if}}\;{l_f}\left( {{g_m},a_i^t} \right) \notin {{\cal L}_{{\rm{G2A}}}} \; {\rm{and}} \; {l_f}\left( {{g_m},s_i^t} \right) \in {{\cal L}_{{\rm{G2S}}}},}\\
					{t_{f,{\rm{SN}}}^{{\rm{non, case2}}}\left( {s_j^t,a_i^t} \right),{\rm{otherwise}}.}
			\end{array}} \right.
		\end{split}
	\end{equation}

	\subsection{Problem Formulation}
	Similar to the case of cached files, we aim to jointly optimize the satellite association $\{\tilde{\bf X}\}$, file download ratio $\{\tilde{\bm{\rho}}\}$, and the bandwidth allocation $\{{\bm{\omega}}\}$, to minimize the total delay of all non-cached files $\{{{t_f^{{\rm{non - cached}}}}\left( {a_i^t} \right)} \}$, which is equivalent to minimize the sum of $\{t_{f}^{\rm relay}\left( {{g_m},s_k^t,s_j^t} \right)\}$. This problem can be formulated as follows.
	\begin{subequations}
		\begin{equation}\label{p5}
			{\mathcal{P}}5:	\mathop {\min }\limits_{{\tilde{\bf X}}, {\tilde{\bm {\rho}}},\boldsymbol{\omega} } \; {\sum\limits_f t_{f}^{\rm relay}\left( {{g_m},s_k^t,s_j^t} \right)} 
		\end{equation}
		\begin{equation}\label{p5_maxISL}
			{\rm{s.t.}} \quad   \sum\limits_{s_k^t \in {{\mathcal S}}} {{\tilde x_f}\left( {s_k^t,s_j^t} \right)}  \le {N_{{\rm{ISL}},\max }}, \; \forall s_j^t \in {\mathcal S},
		\end{equation}
		\begin{equation} \label{p5_rho}
			{\tilde \rho _f}\left( {s_k^t,s_j^t} \right) \le {\tilde x_f}\left( {s_k^t,s_j^t} \right), \; \forall  s_j^t, s_k^t \in {\mathcal S},
		\end{equation}
		\begin{equation}\label{p5_sumrho}
			\sum\limits_{s_k^t \in {{\mathcal S}}} {{\tilde  \rho _f}\left( {s_k^t,s_j^t} \right)}   = 1, \; \forall s_j^t \in {\mathcal S},
		\end{equation}
		\begin{equation}
			{\tilde x_f}\left( {s_k^t,s_j^t} \right) \in \left\{ {0,1} \right\}, \; \forall s_j^t, s_k^t \in {\mathcal S},
		\end{equation}
		\begin{equation}\label{noncache_band}
			\sum\limits_{s_k^t \in {\mathcal S}}\omega \left( {{g_m},{s_k^t}} \right) \leq 1, \; \forall g_m \in {\mathcal G},
		\end{equation}
	\end{subequations}
	where \eqref{noncache_band} ensures the sum of bandwidth allocation factors does not exceed 1. Similar as $\mathcal{P}$1, we omit the subscript $t$ herein-below.

	\subsection{Proposed Efficient Alternating Optimization Algorithm}
	To tackle the above MINLP problem, we propose an efficient algorithm based on alternating optimization. Specifically, we first effectively group the variables and then
	optimize the satellite association and bandwidth allocation iteratively until convergence. The details are given as follows:
	
	\subsubsection{Optimize $\tilde{\bf{X}}$ and $\tilde{\bm{\rho}}$ for given bandwidth allocation $\bm{\omega}$}
	
	As concluded in \textbf{Lemma \ref{lemma1}}, the optimal transmission delay through each established link should be equal. Therefore, the common transmission delay of each file from source node $g_m$ to destination $s_j$ via relay node $s_k$ can be calculated as 
	\begin{equation}
		\begin{split}
			t^*_{f}(g_m,s_k,s_j) & =\frac{{{\tilde \rho _f}\left( {s_k,s_j} \right){b_f}{R_p}}}{{{C}\left( {s_k,s_j} \right)}}+\frac{ {{\tilde \rho _f}\left( {s_k,s_j} \right){b_f}{R_p}} }{C( {g_m,s_k} )}\\
			&=\frac{{b_f}{R_p}}{\sum\limits_{s_k\in\mathcal{S}} \frac{\tilde x_f(s_k,s_j)} {\frac{1}{C(s_k,s_j)}+\frac{1}{C(g_m,s_k)}} }.
		\end{split}
	\end{equation}
	As such, the continuous variable ${\tilde  \rho}_f(s_k,s_j)$ is eliminated. The satellite association subproblem for $\mathcal{P}5$ can be formulated into the following BIP:
	\begin{subequations}
		\begin{align}
			{\mathcal{P}}{\rm{5\text{-}1}}:	\mathop {\min }\limits_{{\tilde{\bf X}}}&\;\; t^*_{f}(g_m,s_k,s_j) \\
			\label{p51_maxISL}{\rm s.t.}  & \sum\limits_{s_k \in {{\mathcal S}}} {{\tilde x_f}\left( {s_k,s_j} \right)}  \le {N_{{\rm{ISL}},\max }}, \; \forall s_j,s_k \in {\mathcal S}, \\
			& {\tilde x_f}\left( {s_k,s_j} \right) \in \{0,1\}, \forall s_j, s_k \in \mathcal{S}. \label{p51_bi}
		\end{align}
	\end{subequations}

	Note that $\mathcal{P}$5-1 has the same form as $\mathcal{P}$2, and hence can be solved via the proposed EPM in Section \ref{sec:cached}. We omit the detailed procedure here.
	
	\subsubsection{Optimize $\bm \omega$  and $\tilde{\bm{\rho}}$ for given task-relay satellite association $\tilde{\mathbf{X}}$} This subproblem can be written as 
	\begin{subequations}
		\begin{align}
			{\mathcal{P}}{\rm{5\text{-}2}}:	\mathop {\min }\limits_{\tilde{\boldsymbol{\rho}},\boldsymbol{\omega}}&\;\; \frac{{{\tilde \rho _f}\left( {s_k,s_j} \right){b_f}{R_p}}}{{{C}\left( {s_k,s_j} \right)}}+\frac{ {{\tilde \rho_f}\left( {s_k,s_j} \right){b_f}{R_p}} }{\omega\left( {{g_m},{s_k}} \right){\log _2} \left( {1 + \frac{c(g_m,s_k)}{\omega( {{g_m},{s_k}} )} }\right) } \\
			{\rm s.t.} \quad  & \; {\tilde \rho _f}\left( {s_k,s_j} \right) \le {\tilde x_f}\left( {s_k,s_j} \right), \; \forall s_j, s_k \in {\mathcal S}, \\
			& \; {\tilde \rho_f}\left( {s_k,s_j} \right) \in [0,1], \; \forall s_j,s_k \in {\mathcal S}, \\
			& \; 	\sum\limits_{s_k \in\mathcal{S}}\omega\left( {{g_m},{s_k}} \right)\leq  1, \; \forall g_m \in \mathcal{G}.
		\end{align}
	\end{subequations}
	
	The difficulty of solving $\cal P$5-2 lies in the second term of the objective function, which is generally non-convex. To tackle it, we give the following lemma, which helps to transform the product of two convex functions into a tractable form.

	\begin{lemma}\label{lemma2}
		Suppose $f(\mathbf{x},\mathbf{y})=f_1(\mathbf{x})f_2(\mathbf{y})$ is the product of two convex and non-negative functions $f_1(\mathbf{x})$ and $f_2(\mathbf{y})$. Then, $f(\mathbf{x},\mathbf{y})$ can be rewritten to have the following difference of convex (DC) structure:
		\begin{equation}
			f(\mathbf{x},\mathbf{y})=\frac{1}{2}\left(f_1(\mathbf{x})+f_2(\mathbf{y}) \right)^2-\frac{1}{2}\left( f_1^2(\mathbf{x})+f_2^2(\mathbf{y})\right).
		\end{equation}
		
		Recall that any convex function is globally lower-bounded by its first-order Taylor expansion at any point.	A convex upper-bound of the above DC structure can be derived using the successive convex approximation as follows:
		\begin{equation}
			\begin{split}
				& \bar{f}(\mathbf{x},\mathbf{y},\mathbf{x}_0,\mathbf{y}_0) \triangleq\frac{1}{2}\left(f_1(\mathbf{x})+f_2(\mathbf{y}) \right)^2\\
				& -\frac{1}{2}\left( f_1^2(\mathbf{x}_0)+f_1^{'}(\mathbf{x}_0)(\mathbf{x}-\mathbf{x}_0)+f_2^2(\mathbf{y}_0)+f_2^{'}(\mathbf{y}_0)(\mathbf{y}-\mathbf{y}_0)\right) \\ 
				& \geq
				f(\mathbf{x},\mathbf{y}),
			\end{split}
		\end{equation}
		where ${\bf{x}}_0$ and ${\bf{y}}_0$ are given local points.
	\end{lemma}
	
	Based on \textbf{Lemma} \textbf{\ref{lemma2}}, the original objective function can be approximated by a tractable form at any feasible point. Specifically, define ${\tilde \rho}^{(r)}$ and $\omega^{(r)}$ as the given download ratio and bandwidth allocation factor in the $r$-th iteration. Therefore, for two convex functions $f_1(\rho)=\rho$, $f_2(\omega)=\frac{1}{\omega\log_2(1+\frac{c_2}{\omega})}$ and local points ${\tilde \rho}^{(r)}$ and $\omega^{(r)}$, a convex upper-bound for $f({\tilde  \rho},\omega)=\frac{\tilde \rho}{\omega\log_2(1+\frac{c_2}{\omega})}$ can be derived by
	\begin{equation}
		\resizebox{1\hsize}{!}{$	\begin{aligned}
				&	f({\tilde \rho},\omega) \\
				& \leq \frac{1}{2}\left[\left( {\tilde \rho}+f_2(\omega) \right)^2-({\tilde \rho}^{(r)}) ^2+2{\tilde \rho}^{(r)}({\tilde \rho}-{\tilde \rho}^{(r)}) \right]   - \frac{1}{2}\left[ {f_2^2\left( {{\omega ^{(r)}}} \right)} \right.\\
				& \left. { - 2f_2^3\left( {{\omega ^{(r)}}} \right)\left( {{{\log }_2}\left( {1 + \frac{c}{{{\omega ^{(r)}}}}} \right) + \frac{{{{\log }_2}\left( e \right)c}}{{{\omega ^{(r)}} + c}}} \right)\left( {\omega  - {\omega ^{(r)}}} \right)} \right]\\
				&\triangleq\bar{f}({\tilde \rho},\omega,{\tilde \rho}^{(r)},\omega^{(r)}).
			\end{aligned} $}
	\end{equation}
	
	Therefore, ${\mathcal{P}}{\rm{5\text{-}2}}$ can be approximated as the following problem:
	\begin{subequations}
		\begin{equation}
			{\mathcal{P}}{\rm{5\text{-}3}}:	\mathop {\min }\limits_{\tilde{\boldsymbol{\rho}},\boldsymbol{\omega}}\;\; \frac{{{{\tilde \rho} _f}\left( {s_k,s_j} \right){b_f}{R_p}}}{{{C}\left( {s_k,s_j} \right)}}+
			b_f R_p \bar{f}({\tilde \rho},\omega,{\tilde \rho}^{(r)},\omega^{(r)})
		\end{equation}
		\begin{equation}
			{\rm s.t.}  \quad {{\tilde \rho} _f}\left( {s_k,s_j} \right) \le {{\tilde x}_f}\left( {s_k,s_j} \right), \; \forall  s_j,s_k \in {\mathcal S},
		\end{equation}
		\begin{equation}
			{\tilde \rho_f}\left( {s_k,s_j} \right) \in [0,1], \; \forall s_j,s_k \in {\mathcal S},
		\end{equation}
		\begin{equation}
			\sum\limits_{s_k^t \in\mathcal{S}}\omega\left( {{g_m},{s_k^t}} \right)\leq  1, \; \forall g_m \in \mathcal{G}.
		\end{equation}
	\end{subequations}

	\begin{algorithm}[!t]
		\caption{Proposed alternating optimization for solving $\mathcal{P}$5}
		\begin{algorithmic}[1]\label{alg2}
			\State {Initialize iteration number $r=0$, bandwidth allocation $\bm{\omega}^{(0)}$, and convergence threshold $\xi$.}
			\State \quad {\bf repeat}   
			\State \quad\quad  Solve ${\mathcal{P}}{\rm{5\text{-}1}}$ by EPM.
			\State \quad\quad  Solve  ${\mathcal{P}}{\rm{5\text{-}3}}$ until convergence.
			\State \quad\quad  Update the penalty factor.
			\State \quad\quad $r=r+1$.
			\State \quad {\bf until} The decrease in the objective function of $\mathcal{P}$5 is less than $ \xi$.
		\end{algorithmic}
	\end{algorithm}
	
	Since the objective function is convex and all the constraints are affine, problem ${\mathcal{P}}{\rm{5\text{-}3}}$ is a convex optimization problem. The procedure for solving $\mathcal{P}5$ is summarized in Algorithm \ref{alg2}.

	\subsection{Complexity Analysis}
	In each iteration of Algorithm \ref{alg2}, solving the convex problems ${\mathcal{P}}{\rm{5\text{-}1}}$ and ${\mathcal{P}}{\rm{5\text{-}3}}$ exhibit computational complexity that is merely polynomial in terms of the number of variables and constraints. We also consider the worse case, where all the aircraft have content requests and their associated satellite has $N_{\rm{ISL,\max}}$ visible neighboring satellites. 
	In this case, for ${\mathcal{P}}{\rm{5\text{-}1}}$, there is a linear objective, $AN_{\rm{ISL,\max}}$ real-valued variables, and $A+AN_{\rm{ISL,\max}}$ linear constraints. Thus, the computational complexity of ${\mathcal{P}}{\rm{5\text{-}1}}$ is upper bounded $\mathcal{O}\left( {\left( {A + 2A{N_{{\text{ISL}},\max }}} \right){{\left( {A{N_{{\text{ISL}},\max }}} \right)}^2}\sqrt {A + A{N_{{\text{ISL}},\max }}} } \right)$. For ${\mathcal{P}}{\rm{5\text{-}3}}$, the number of real-valued variables is at most $A+AN_{\rm{ISL,\max}}$, and the number of linear constraints is at most $A+2AN_{\rm{ISL,\max}}$. Thus, the computational complexity of ${\mathcal{P}}{\rm{5\text{-}3}}$ is upper bounded by $ \mathcal{O}\left( {\left( {2A + 3A{N_{{\text{ISL}},\max }}} \right){{\left( {A{N_{{\text{ISL}},\max }}} \right)}^2}\sqrt {A + 2A{N_{{\text{ISL}},\max }}} } \right) $. In addition, the complexity of updating the penalty factor is upper bounded by $\mathcal{O}\left(1\right)$.

	\section{Simulation Results}\label{sec_sr}
	
	This section provides numerical results to validate the effectiveness of our proposed transmission schemes for cached and non-cached files. 
	
	\subsection{Simulation Setups}
	We consider a whole day which is divided into multiple TSs with 15 seconds. A total of 120 satellites are distributed across 6 orbital planes, with each plane containing 20 satellites at an altitude of 1000 km, and the orbit inclination is set at 53$^{\circ}$. We consider Airbus A320 aircraft, each of which can generate at most one file request in a TS. Following the service configurations outlined in \cite{9810267}, we take into account four distinct types of files: music, image, video (these three types can be proactively cached in the satellite network), and instantaneous data stream (which can only be derived from the GSs). The number of packets follows a uniform distribution within the ranges of [50, 100], [500, 1000], [1000, 3000], and [10, 1000], respectively. Each packet contains 1080 bits. The default number of maximum ISLs per satellite and the number of GSs are 2 and 5, respectively. Other simulation parameters are given in Table \ref{sim_parameters} unless stated otherwise \cite{9810267,trends}.
	
	\begin{table}[ht]
		\caption{MAIN SIMULATION PARAMETER VALUES}\label{sim_parameters}
		\vspace{-15pt}
		\begin{center}
			\begin{tabular}{|m{4cm}<{\centering}|c|}
				\hline
				Parameter      & Value  \\ \hline
				Center frequency of S2A, G2A, and G2S links & 15 GHz, 18 GHz, 30 GHz\\ \hline
				Bandwidth of S2A G2A, and G2S links & 100 MHz\\ \hline
				Antenna diameter of satellite (S2A link), aircraft and GS & 1 m, 0.5 m, 2 m \\  \hline
				Corresponding antenna gain of satellite, aircraft and GS & 40 dB, 30 dB, 52 dB \\  \hline
				Transmit power of satellite and GS & 5 W, 10 W \\ \hline
				Carrier frequency, bandwidth and antenna gain of laser ISL & 193 THz, 50 MHz, 90 dB \\ \hline
				Additional loss in the space and air & 5.2 dB, 2.5 dB \\ \hline
				EPM-related parameters and convergence threshold $\epsilon$, $\Delta$, $\xi$ & $10^{-4}$, 5, $10^{-5}$ \\ \hline
			\end{tabular}
		\end{center}
	\end{table}
	
	\subsection{Performance Comparison}
	
	First, to evaluate the performance of the proposed EPM for the cached files, we choose the following benchmarks: 
	\begin{enumerate}
		\item \textit{Exhaustive search-based method}: In this case, the optimal solution is obtained by searching over all possibilities of satellite association in $\bf X$. 
		
		\item \textit{Greedy-based satellite association}: In this case, for all satellites that receive file requests, the ISL with the maximum capacity is consecutively selected until the maximum number of ISLs is reached.
		
		\item \textit{Random satellite association}: In this case, the connection status of all satellites are randomly generated. 
	\end{enumerate}

	Fig. \ref{fig:convergence} shows the convergence of the proposed algorithms. We can find that the proposed algorithms can achieve convergence within a limited number of iterations. Then, in Fig. \ref{cached_vs_task}, we plot the average delay of each file versus the average number of (cached) file requests per TS. Several key observations can be made as follows. Firstly, the performance of  Algorithm \ref{alg1} approximates that of the exhaustive search-based scheme, which verifies the effectiveness of the proposed EPM. Secondly, when the number of file requests is relatively small (i.e., less than 20), the greedy-based satellite association scheme achieves close performance as the optimal solution, while the performance gap increases as the requests grow intense. This phenomenon is consistent with the analysis in Remark \ref{remark2}. Thirdly, the random-based scheme maintains a high delay and thus is not an appropriate choice for satellite association.

    	\begin{figure}[h]
		\centering
		\includegraphics[width = 0.45\textwidth]{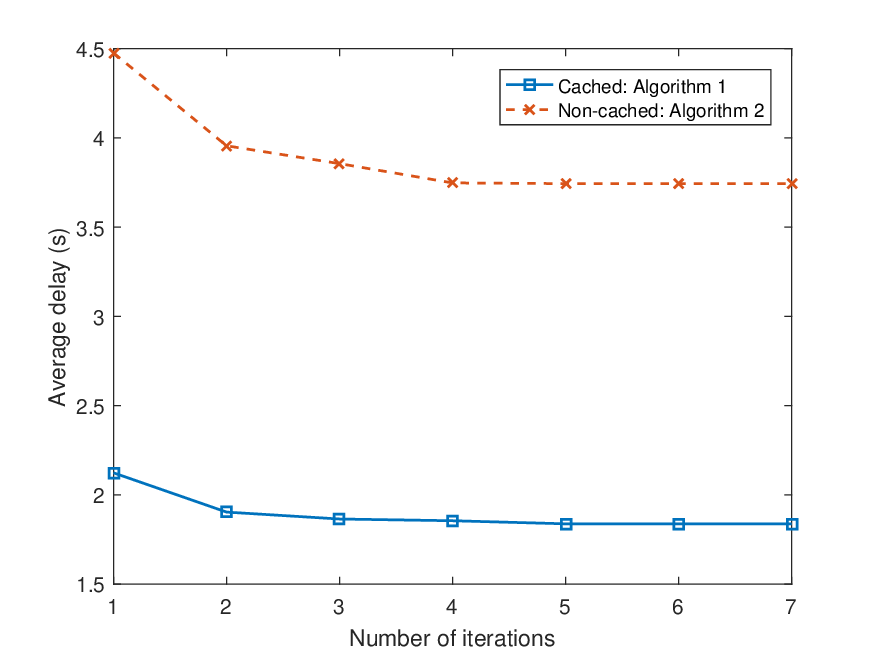}
		\caption{Convergence of the proposed Algorithm 1 and 2.}\label{fig:convergence}
	\end{figure}
    
	\begin{figure}[h]
		\centering
		\includegraphics[width = 0.45\textwidth]{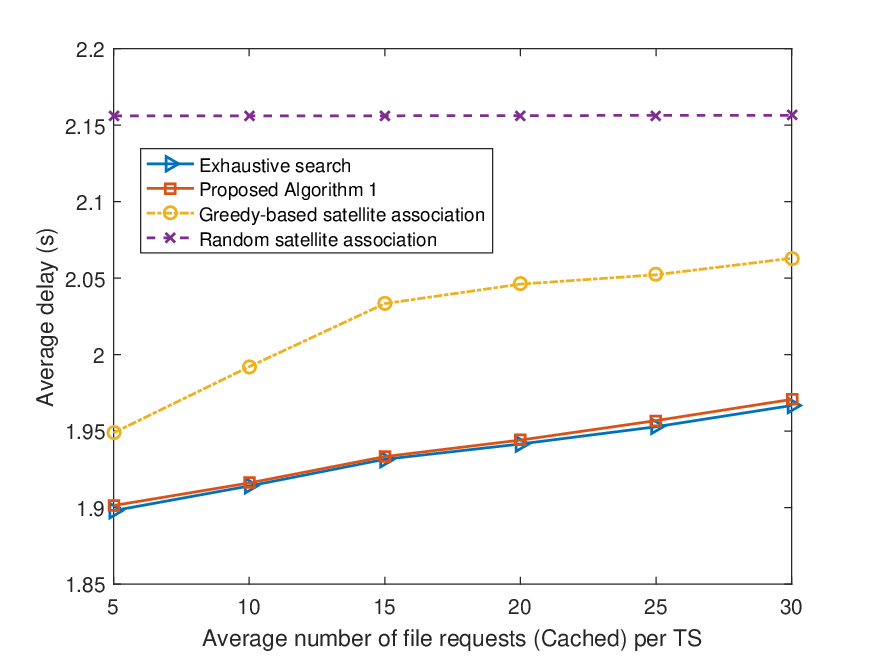}
		\caption{Average delay versus the average number of (Cached) file requests per TS.}\label{cached_vs_task}
	\end{figure}

	Then, for the non-cached files, the following benchmark schemes are chosen: 
	\begin{enumerate}
		\item \textit{Exhaustive search-based method}: In this case, $\cal P$5 is optimally solved by invoking the branch and bound method, whose performance can serve as a lower-bound for the average delay. 
		\item \textit{Equal bandwidth allocation}. This scheme equally allocates the total bandwidth among all GSs.
		\item \textit{Rounding-based satellite association}: In this case, we first relax the variable $x_f(s_k,s_j)$ into the range of $[0,1]$; Then, we solve $\cal P$5-1 with the relaxed $x_f(s_k,s_j)$ and round it into $\{0,1\}$  to obtain the satellite association; Finally, the bandwidth allocation is optimized. 
	\end{enumerate}

	Fig. \ref{noncached_vs_task} shows the average delay versus the number of (non-cached) file requests per TS. Our proposed Algorithm \ref{alg2} outperforms the equal bandwidth allocation and rounding-based association schemes while achieving close performance as the exhaustive search-based scheme. This validates the importance of a comprehensive satellite association and bandwidth allocation design in the considered IFC-oriented SAGIN.

    \begin{figure}[h]
		\centering
		\includegraphics[width = 0.45\textwidth]{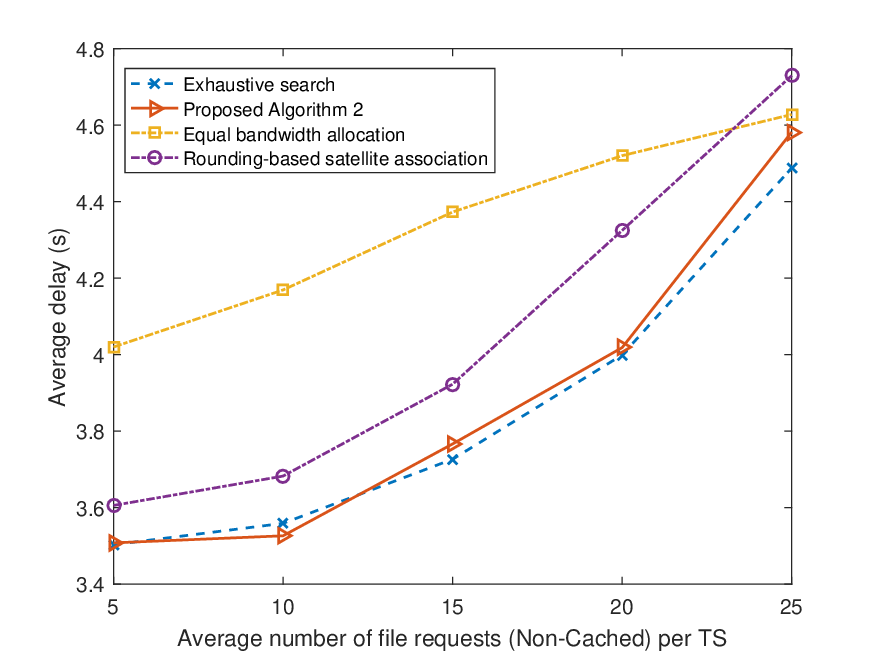}
		\caption{Average delay versus the average number of (Non-cached) file requests per TS.}\label{noncached_vs_task}
	\end{figure}

	\subsection{Impact of Several Key Parameters}
	\subsubsection{Impact of the maximum number of ISLs} Fig. \ref{vs_ISL} depicts the average delay versus the maximum number of realizable ISLs. The delay when neglecting the constraint of ISLs (i.e., fully connected satellite network) is also illustrated. For both cached and non-cached files, the average delay for downing the requested files decreases dramatically with the maximum number of ISLs. This is attributed to the high capacity provided by the laser ISLs, which significantly reduces the transmission delay within the satellite networks. Compared with permanent and single-path ISLs, we can conclude that dynamic and multi-path ISLs will play a vital role in achieving low-latency content delivery for next-generation wireless satellite networks.
	Last but not least, the delay will not monotonically decrease with the number of maximum ISLs as it converges to that of the fully connected scheme. This is because the number of visible satellites that have cached the request files is limited. Note that the trade-off between the cost of realizing more ISLs and the performance gain it brings is still an open and important issue.

    	\begin{figure}[h]
		\centering
		\includegraphics[width = 0.45\textwidth]{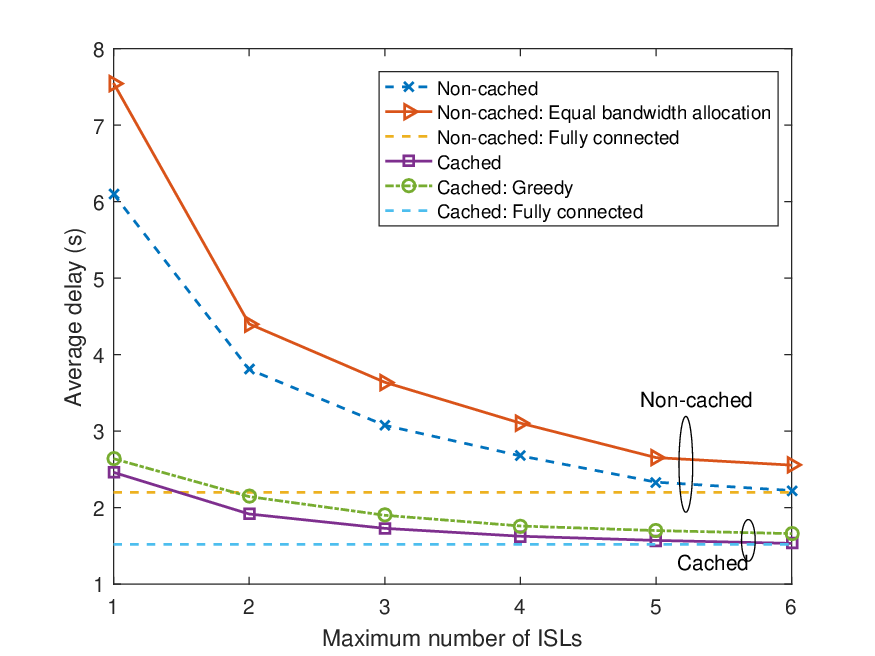}
		\caption{Average delay versus the maximum number of realizable ISLs.}\label{vs_ISL}  
	\end{figure}
	
	\subsubsection{Impact of the number of GSs} Fig. \ref{vs_GS} illustrates the average delay versus the number of GSs. The average delay decreases with the number of GSs for both types of files. However, the delay reduction is prominent for the non-cached files, while the decrease in delay with the bandwidth is relatively slight for the cached files. This is because the capacity of G2S links generally serves as a bottleneck for the non-cached files. In contrast, satellite GSs only act as a supplementary means for content delivery for the cached files. Moreover, the file delivery delay of non-cached files is more significant than that of the cached files since they require a multi-hop transmission. This result verifies the importance of deploying more gateways/GSs, especially for the non-cached files.
	
	\begin{figure}[h]
		\centering
		\includegraphics[width = 0.45\textwidth]{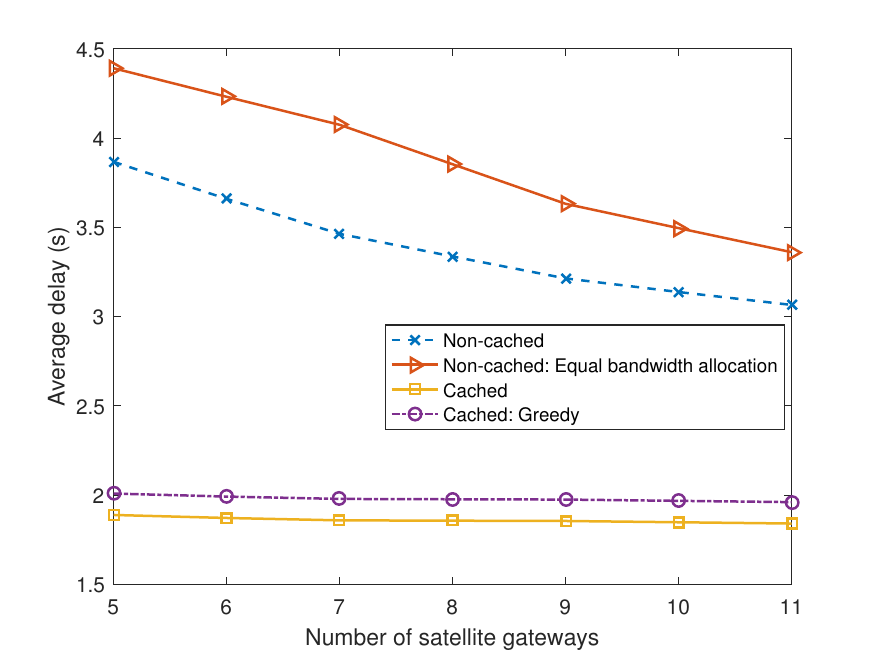}
		\caption{Average delay versus the number of GSs.}\label{vs_GS}
	\end{figure}
    
	\subsubsection{Impact of the bandwidth of satellite gateways and the altitudes of LEOs} Finally, we plot the average delay versus the total bandwidth of each satellite with different altitudes of LEOs in Fig. \ref{vs_band}. It can be observed that reducing the altitudes of the satellites can decrease the transmission delay of the non-cached files, while the decrease is slight for the cached files.

	\begin{figure}[h]
		\centering
		\includegraphics[width = 0.45\textwidth]{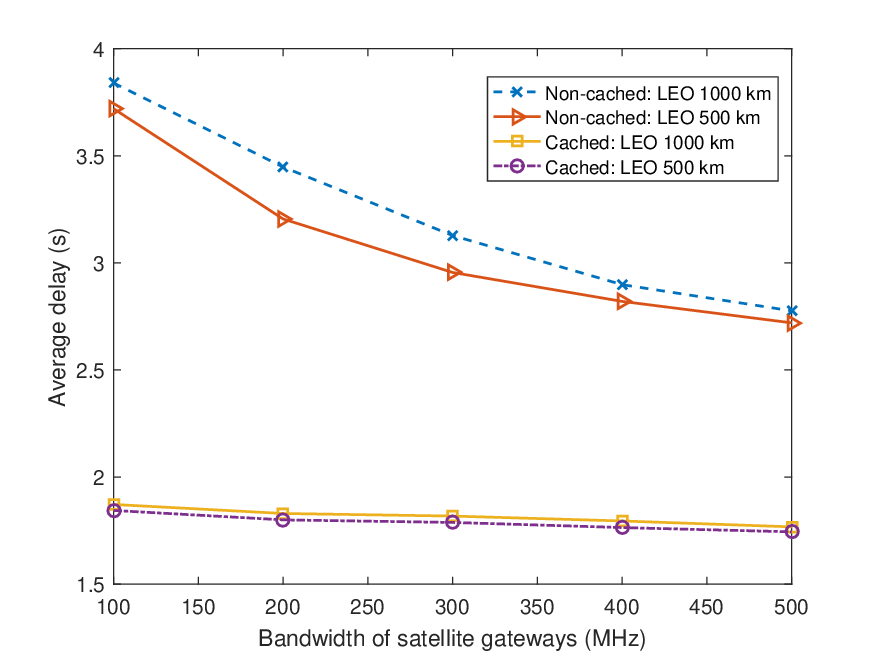}
		\caption{Average delay versus the bandwidth of GSs and the altitudes of LEO satellites.}\label{vs_band}  
	\end{figure}

	\section{Conclusion}\label{sec_conclusion}
	To improve the communication experience of airborne passengers, this article investigated the in-flight connectivity optimization problem in SAGIN, where satellites and GSs collaboratively deliver content to aircraft. By utilizing high-speed laser ISLs, we proposed a novel dual delivery framework in IFC-oriented SAGIN to minimize the average downloading delay and developed the corresponding optimization algorithms. The following insights were gained for practical IFC scheme design in the context of aviation IoT: First, the fluctuation of non-cached files' transmission latency is more sensitive than that of cached contents, as the communication resource of the terrestrial network serves as a bottleneck. Second, the maximum number of ISLs is crucial for reducing transmission latency, although the delay reduction converges as the number of ISLs increases.
	
	This article provided a preliminary attempt at optimizing content delivery in IFC-oriented SAGIN. 
	Several promising research directions remain for future exploration:
	\begin{itemize}
		\item Evaluate the performance of IFC schemes using comprehensive performance metrics for multimedia applications, including delay jitter and packet loss rate.
		
		\item Devise methods for mitigating the effects of misalignment fading when applying RISs on small satellites.
		
		\item 
		Design adequate frequency compensation to mitigate the effects of Doppler shifts, thereby ensuring the stability of data decoding and transmission. 
	\end{itemize}

\section*{Appendix A}	
First, we prove that $\mathbf{u}\in \{-1,1\}^{m\times n}$. According to the Cauchy-Schwarz inequality, we have 
\begin{equation}
	\langle \mathbf{u},\mathbf{v}\rangle=mn \leq \Vert\mathbf{u}\Vert_2 \cdot \Vert\mathbf{v}\Vert_2 \leq \sqrt{mn}\Vert\mathbf{u}\Vert_2.
\end{equation}
Then we have $\Vert\mathbf{u}\Vert_2^2\geq mn$. Considering that $-1\preceq {\bf u} \preceq 1$, $\Vert\mathbf{u}\Vert_2^2\leq mn$ also holds. Therefore, we have $\Vert\mathbf{u}\Vert_2^2= mn$ and $\mathbf{u}\in \{-1,1\}^{m\times n}$.

Next, we prove that $\mathbf{v}\in \{-1,1\}^{m\times n}$. The following inequality can be derived
\begin{equation}\label{theorem1_2}
	\begin{aligned}
		\langle \mathbf{u},\mathbf{v}\rangle=mn =&\sum_{i=1}^{m}\sum_{j=1}^{n}u_{ij}v_{ij}\leq
		\sum_{i=1}^{m}\sum_{j=1}^{n}|u_{ij}||v_{ij}| \\ \leq& \sum_{i=1}^{m}\sum_{j=1}^{n}|v_{ij}| \leq \sqrt{mn}\Vert\mathbf{u}\Vert_2. 
	\end{aligned}
\end{equation}
Thus, $\Vert\mathbf{v}\Vert_2^2\geq mn$. Combining that $\Vert\mathbf{v}\Vert_2^2\leq mn$, we have $\Vert\mathbf{v}\Vert_2^2= mn$. With the Squeeze theorem, all the equalities in \eqref{theorem1_2} hold automatically. Using the equality condition for the Cauchy-Schwarz inequality, we obtain $\mathbf{v}\in \{-1,1\}^{m\times n}$.

Given that $\mathbf{u},\mathbf{v}\in \{-1,1\}^{m\times n}$ and $\langle \mathbf{u},\mathbf{v}\rangle=mn$, we can obtain $\mathbf{u}=\mathbf{v}$. This complete the proof.

	\ifCLASSOPTIONcaptionsoff
	\newpage
	\fi

	\bibliographystyle{IEEEtran}
	\bibliography{IEEEabrv,reference}

\end{document}